%% file: NewRulesForTBR-Accepted_Version.tex
\documentclass[11pt]{llncs}

\usepackage{color}
\usepackage{url}

\usepackage[usenames,dvipsnames]{xcolor}
\usepackage{multirow}
\usepackage{graphics}
\usepackage{a4wide}
\usepackage{pdfpages}
\usepackage{enumerate}

\newcommand{\TBR}{{\rm TBR}}

\newcommand{\MAF}{{\rm MAF}}

\newtheorem{observation}[theorem]{Observation}

\newcommand{\blueblue}[1]{\textcolor{black}{#1}}

\newcommand{\blue}[1]{\textcolor{black}{#1}}

\newcommand{\newblue}[1]{\textcolor{black}{#1}}

\definecolor{applegreen}{rgb}{0.55, 0.71, 0.0}
\definecolor{antiquefuchsia}{rgb}{0.57, 0.36, 0.51}

\newcommand{\polish}[1]{\textcolor{black}{#1}}

\newcommand{\green}[1]{\textcolor{black}{#1}}
\newcommand{\fuchsia}[1]{\textcolor{black}{#1}}

\newcommand{\redred}[1]{\textcolor{black}{#1}}

\usepackage{amsmath}
\usepackage{graphicx}

\title{New reduction rules for the tree bisection and reconnection distance}
\author{Steven Kelk\inst{1}, Simone Linz\inst{2}}

\institute{Department of Data Science and Knowledge Engineering (DKE),\\ Maastricht University, The Netherlands,\\ \email{steven.kelk@maastrichtuniversity.nl}
\and \fuchsia{School} of Computer Science, University of Auckland, New Zealand,\\
\email{s.linz@auckland.ac.nz}}

\providecommand{\keywords}[1]{\textit{Keywords:} #1}

\begin{document}
\maketitle

\begin{abstract}
Recently it was shown that, if the subtree and chain reduction rules have been applied exhaustively to two
unrooted phylogenetic trees, the reduced trees will have at most $15k-9$ taxa where $k$ is the TBR (Tree Bisection and Reconnection) distance between the two trees, and that this bound is tight. Here we propose five new reduction
rules and show that these further reduce the bound to $11k-9$. The new rules combine the ``unrooted
generator'' approach introduced in \cite{tightkernel} with a careful analysis of agreement forests to identify (i) situations when chains of length 3 can be further shortened without reducing the TBR distance, and (ii) situations when \green{small subtrees} can be identified whose deletion is guaranteed to reduce the TBR
distance by 1. To the best of our knowledge these are the first reduction rules that strictly enhance the reductive power of the subtree and chain reduction rules.  
\end{abstract}


\section{Introduction}
A phylogenetic tree is a tree whose leaves are bijectively labelled by a set of species (or, more generically, a set of \emph{taxa}) $X$ \cite{SempleSteel2003}. These trees are ubiquitous in the systematic study of evolution: \fuchsia{the} leaves represent contemporary species and the internal vertices of the tree represent hypothetical common ancestors. Over the years many techniques have been developed for inferring phylogenetic trees from (incomplete) biological data and under a range of different objective functions \cite{felsenstein2004inferring}. Here we are not concerned with inferring phylogenetic trees, but rather with quantifying the ``distance'' between two phylogenetic trees. Such a goal is well-motivated, since different methodologies for inferring phylogenetic trees sometimes yield trees with differing topologies, and reticulate evolutionary phenomena such as hybridization can cause different genes in the same genome to have different evolutionary histories \cite{HusonRuppScornavacca10}.\\

We focus on the Tree Bisection and Reconnection (TBR) distance, which is NP-hard to compute \cite{AllenSteel2001,hein1996complexity}. Informally, the TBR distance between two trees $T$ and $T'$, denoted $d_\TBR(T,T')$, is the minimum number of topological rearrangement moves that need to be applied to transform $T$ into $T'$, where such a move involves detaching a subtree and attaching it elsewhere. It was proven in 2001 \cite{AllenSteel2001} that the question ``Is $d_\TBR(T,T') \leq k?$'' can be answered in time $f(k) \cdot \text{poly}(|X|)$, where $f$ is a computable function that depends only on $k$. In other words: the problem is \emph{fixed parameter tractable} \cite{cygan2015parameterized}. Specifically, the authors proved that the \fuchsia{two} polynomial-time \emph{subtree} and \emph{chain} reduction rules preserve the TBR distance and reduce the number of taxa \fuchsia{to at most $28 \cdot d_\TBR(T,T')$ for any two unrooted phylogenetic trees $T$ and $T'$. The} reduced instance, known as a \emph{kernel}, can then be solved with any exact algorithm, yielding the $f(k) \cdot \text{poly}(|X|)$ running time \cite{kernelization2019}. The analysis in  \cite{AllenSteel2001} made  heavy use of a powerful abstraction known as an \emph{agreement forest}, which in a nutshell partitions the two trees into smaller, non-overlapping fragments which do have the same topology in both trees (see e.g. \cite{shi2018parameterized,whidden2013fixed} for overviews of algorithmic results and \cite{extremal2019} for extremal results). Via a different technique, based on bounded search trees,  running times of $O( 4^k \cdot \text{poly}(|X|) )$ \cite{whidden2013fixed} and
then $O( 3^k \cdot \text{poly}(|X|) )$ were later obtained \cite{chen2015parameterized}. However, the question remained whether a kernel with fewer than $28 \cdot d_\TBR(T,T')$ taxa could be obtained.

Recently, in \cite{tightkernel}, it was shown that the subtree and chain reduction rules actually reduce the instance to size $15\cdot d_\TBR(T,T')-9$, and that there are instances for which this bound is tight. Interestingly, the sharpened analysis does not leverage agreement forests at all, but instead recasts the computation of TBR distance as a phylogenetic \emph{network} inference problem, where phylogenetic networks are essentially the generalization of phylogenetic trees to graphs. Namely, the TBR distance of $T$ and $T'$ is equal to the minimum value of $|E|-(|V|-1)$, ranging over all phylogenetic networks $N=(V,E)$ that embed $T$ and $T'$ \cite{van2018unrooted}. The backbone topology of such minimal networks can be represented by \emph{unrooted generators}, and when viewed from this static perspective it becomes much easier to \fuchsia{analyze} the role of common chains in the trees.

In the present article we combine the agreement forest perspective of \cite{AllenSteel2001} with the network/generator perspective of \cite{tightkernel}, to obtain a new suite of five polynomial-time reduction rules. When applied alongside the subtree and chain reduction rules, these reduce the size of the kernel to $11\cdot d_\TBR(T,T')-9$. To leverage agreement forests, we first prove a general theorem which states that, given any disjoint set of common chains, there exists an optimal agreement forest in which all the chains are simultaneously preserved i.e. none of the chains are divided across two or more components of the forest. Crucially, this also holds for chains containing only 2 taxa, as long as the two taxa in the chain have a common parent in at least one of the trees. Such very short chains have not received a lot of attention in the literature, since the standard chain reduction \fuchsia{rule, which truncates long common chains to length 3, does not always preserve the TBR distance} if the chains are truncated to length 2. Nevertheless, the weaker ``simultaneous preservation'' property that we prove in this article \fuchsia{(see Theorem~\ref{thm:allchainsintact})}, \green{and which we believe to be of independent interest}, turns out to be quite powerful when combined with networks/generators. The fact that chains are preserved allows us to determine specific situations when it \emph{is} actually safe to reduce a chain to length 2 \green{(and sometimes to length 1)}, or even to identify an entire component of an optimal agreement forest (which can then be deleted, reducing the TBR distance by exactly 1). These insights directly inspire the new reduction rules presented in this article. To the best of our knowledge these new reduction rules are the first reduction rules for a phylogenetic distance problem which strictly improve upon the reductive power of the subtree and chain reduction rules. Other reduction rules, such as the \emph{cluster} reduction \fuchsia{\cite{baroni2006hybrids,bordewich2017fixed}}, tend to be very effective in practice, but do not yield improved (i.e. smaller) bounds on kernel size \cite{tightkernel}.

After presenting the main results, we show a family of tight examples i.e. tree pairs that, after applying all \fuchsia{seven} reduction rules, have exactly $11\cdot d_\TBR(T,T')-9$ taxa. We then conclude with a short reflection on potential avenues for further improving the $11\cdot d_\TBR(T,T')-9$ bound, and discuss a number of insights flowing from our analysis which might be useful when considering non-kernelization approaches for computing \fuchsia{the} TBR distance.

\section{Preliminaries}\label{sec:prelim}

Throughout this paper, $X$ denotes a finite set \polish{(of \emph{taxa})} with $|X|\geq 4$. \\

\noindent{\bf Unrooted phylogenetic trees and networks.} An {\it unrooted binary phylogenetic network} $N$ on $X$  is a  simple, connected, and undirected graph whose leaves are bijectively labeled with $X$ and whose other vertices \green{all have} degree 3.  Let $E$ and $V$ be the edge and vertex set of $N$, respectively. We define the {\it reticulation number} of $N$ as the number of edges in $E$ that need to be deleted from $N$ to obtain a spanning tree. More formally, we have $r(N) = |E|-(|V|-1)$. If $r(N)=0$, then $N$ is called an {\it unrooted binary phylogenetic tree} on $X$. An example of two unrooted binary phylogenetic trees is shown in Figure~\ref{fig:trees}. Now, let $T$ be an unrooted binary phylogenetic tree on $X$.  Two leaves, say $a$ and $b$, of $T$ are called a {\it cherry} $\{a,b\}$ of $T$ if they \fuchsia{are adjacent to a common vertex. We say that a vertex $v$ is the (unique) {\it parent} of a leaf $a$ in $N$ if $v$ is adjacent to  $a$.} For $X' \subset X$, we write $T[X']$ to denote the unique, minimal subtree of $T$ that connects all elements in $X'$. For brevity we call $T[X']$ the \emph{embedding} of  $X'$ in $T$. 
\newblue{If $X''$ is also a subset of $X$, we denote by $T[X']\cap T[X'']$ the set of \textcolor{black}{vertices} in $T$ that are contained in $T[X']$ and $T[X'']$.}
Furthermore, we refer to the unrooted phylogenetic tree on $X'$ obtained from $T[X']$ by suppressing degree-2 vertices as  the {\it restriction of $T$ to $X'$}  \green{and we denote this by} $T|X'$.  \\

\begin{figure}[t]
\center
\scalebox{1}{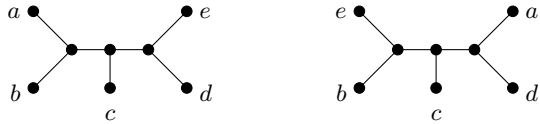}
\caption{\blue{Two unrooted binary phylogenetic trees on $X=\{a,b,c,d,e\}$.}}
\label{fig:trees}
\end{figure}

Let $T$ be an unrooted binary phylogenetic tree on $X$. A {\it quartet} is an unrooted binary phylogenetic tree with exactly four leaves. For example, if $\{a,b,c,d\}\subseteq X$ , we say that $ab|cd$ is a quartet of $T$ if  the path from $a$ to $b$ does not intersect the path from $c$ to $d$. Note that, if $ab|cd$ is not a quartet of $T$, then either $ac|bd$ or $ad|bc$ is a quartet of $T$.  \\

Lastly, let $N$ be an unrooted binary phylogenetic network on $X$ and let $T$ be an unrooted binary phylogenetic tree on $X$. We say that $N$ {\it displays} $T$, if  $T$ can be obtained from a subtree of $N$ by suppressing degree-2 vertices. \\


\noindent{\bf Tree bisection and reconnection.} Let $T$ be an unrooted binary phylogenetic tree on \green{$X$}. Apply the following three-step operation to $T$:
\begin{enumerate}
\item Delete an edge in $T$ and suppress any resulting degree-2 vertex. Let $T_1$ and $T_2$ be the two resulting unrooted binary phylogenetic trees.
\item If $T_1$ (resp. $T_2$) has at least one edge, subdivide an edge in $T_1$ (resp. $T_2$) with a new vertex $v_1$ (resp. $v_2$) and otherwise set $v_1$ (resp. $v_2$) to be the single isolated vertex of $T_1$ (resp. $T_2$).
\item Add a new edge \polish{$\{v_1,v_2\}$} to obtain a new unrooted binary phylogenetic tree $T'$ on $X$.
\end{enumerate}
We say that $T'$ has been obtained from $T$ by a single  {\it tree bisection and reconnection} (TBR) operation. \blueblue{An example of a TBR operation is illustrated in Figure~\ref{fig:tbr}.} Furthermore, we define the TBR {\it distance}  between two  unrooted binary phylogenetic trees $T$ and $T'$ on $X$,  denoted by $d_\TBR(T,T')$, to be the minimum number of TBR operations that is required to transform $T$ into $T'$. It is well known that $d_\TBR$ is a metric~\cite{AllenSteel2001}. By building on an earlier result by Hein et al.~\cite[Theorem 8]{hein1996complexity}, Allen and Steel~\cite{AllenSteel2001} showed that computing the TBR distance is an NP-hard problem.\\

\begin{figure}[t]
\center
\scalebox{1}{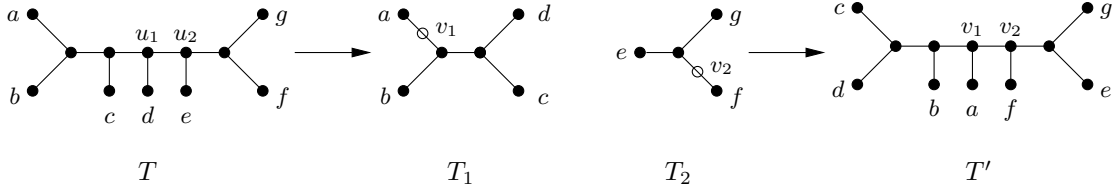}
\caption{\blueblue{A single TBR operation that transforms $T$ into $T'$. First, $T_1$ and $T_2$ are obtained from $T$ by deleting the edge $\{u_1,u_2\}$ in $T$. Second, $T'$ is obtained from $T_1$ and $T_2$ by subdividing an edge in both trees as indicated by the open circles $v_1$ and $v_2$ and adding a new edge $\{v_1,v_2\}$.}}
\label{fig:tbr}
\end{figure}

\noindent{\bf \blue{Unrooted} minimum hybridization.} In \cite{van2018unrooted}, it was shown that computing the TBR distance for a pair of unrooted binary phylogenetic trees $T$ and $T'$ on $X$ is equivalent to computing the minimum number of \polish{extra edges} required to simultaneously explain $T$ and $T'$. More precisely, we set
$$h^u(T, T') = \min_ N\{r(N)\},$$ where the minimum is taken over all unrooted binary phylogenetic networks $N$ on $X$ that display $T$ and $T'$. \polish{The value $h^u(T, T')$ is known as the \emph{(unrooted) hybridization number} of $T$ and $T'$ \cite{van2018unrooted}.} \\


\noindent The aforementioned equivalence is given in the next theorem that was established in~\cite[Theorem 3]{van2018unrooted}.

\begin{theorem}\label{t:tbr-equiv}
Let $T$ and  $T'$ be two unrooted binary phylogenetic trees on $X$. Then $$d_\TBR(T,T')=h^u(T,T').$$
\end{theorem}


\noindent{\bf Subtrees and chains.} Let $N$ be an unrooted binary phylogenetic network on $X$. A {\it pendant subtree} of $N$ is an unrooted binary phylogenetic tree on a proper subset of $X$ that can be obtained from $N$ by deleting a single edge. For $n\geq 1$, let $C=(\ell_1,\ell_2,\ldots,\ell_n)$ be a sequence of distinct taxa in $X$ and, for each $i\in\{1,2,\ldots,n\}$, let $p_i$ denote the unique \fuchsia{parent} of $\ell_i$ in $N$. We call $C$ an {\it $n$-chain} of $N$, where $n$ is referred to as the {\it length} of $C$,  if there exists a walk $p_1,p_2,\ldots,p_n$ in $N$ and the elements in $\{p_2,p_3,\ldots,p_{n-1}\}$ are all pairwise distinct. Hence  $\ell_1$ and $\ell_2$ may have a common parent \green{(i.e. $p_1 = p_2$)} and/or $\ell_{n-1}$ and $\ell_{n}$ may have a common parent in $N$ \green{(i.e. $p_{n-1} = p_n$)}.  \green{Note that, if one of $p_1 = p_2$ and $p_{n-1}=p_n$ holds, then $C$ is pendant in $N$.} Since we require that $X$ contains at least four elements, note that a 3-chain of $N$ cannot consist of three leaves that all have the same parent. Furthermore, by definition, each element in $X$ is a chain of length 1 in $N$. To ease reading, we sometimes write $C$ to denote the set $\{\ell_1,\ell_2,\ldots,\ell_n\}$. It will always be clear from the context whether $C$ refers to the associated sequence or set of leaves.
\\

If a pendant subtree $S$ (resp. chain $C$) exists in two unrooted binary phylogenetic trees $T$ and $T'$, we say that $S$ (resp. $C$) is a \emph{common subtree} \blueblue{(resp. \emph{chain})} of $T$ and $T'$. To illustrate, the 3-chain $(b,c,d)$ is a common chain of the two phylogenetic trees shown in Figure~\ref{fig:trees}. Moreover, if $C$ is a common \blueblue{$n$-chain} of $T$ and $T'$, \blueblue{\emph{reducing $C$ to a chain of  length $k$ with $1\leq k< n$}} yields the two new trees $T_r = T|X \setminus \{ \ell_{k+1}, \ldots, \ell_{n} \}$ and $T'_r = T'|X \setminus \{ \ell_{k+1}, \ldots, \ell_{n} \}$.\\

We will later make use of the following simple, but fundamental, observation, which holds due to the definition of an $n$-chain and the fact that we are working with unrooted binary trees.

\begin{observation}
\label{obs:disjoint}
\fuchsia{Let $T$ be an unrooted binary phylogenetic tree  on $X$, and let $\{C_1,C_2,\ldots,C_m\}$ be a set of mutually taxa-disjoint chains of $T$. Then the embeddings $$\{T[C_1],T[C_2],\ldots,T[C_m]\}$$  are mutually vertex disjoint in $T$.}
\end{observation}


For two unrooted binary phylogenetic trees $T$ and $T'$, we next state two \fuchsia{reduction rules} that were first introduced in~\cite{AllenSteel2001} and crucial in establishing fixed-parameter tractability of computing the TBR distance (see Section~\ref{sec:new-red} for details).\\

\noindent {\bf Subtree reduction.} Replace a maximal pendant subtree with at least two leaves that is common to $T$ and $T'$ by a single leaf with a new label.\\

\noindent {\bf Chain reduction.} Reduce a maximal $n$-chain with $n\geq 4$ that is common to $T$ and $T'$ to a chain of length three.\\

\noindent The next lemma shows that the subtree and chain reduction \fuchsia{are both \emph{TBR-preserving}}, \blueblue{i.e. applying one of the two reductions to $T$ and $T'$ results in a pair of new trees whose TBR distance is equal to $d_\TBR(T,T')$}.

\begin{lemma}\label{t:safe}
Let $T$ and $T'$ be two  unrooted binary phylogenetic trees on $X$, and let $T_r$ and $T_r'$ be two trees obtained from $T$ and $T'$, respectively, by applying a single subtree or chain reduction. Then $d_\TBR(T,T')=d_\TBR(T_r,T_r')$.
\end{lemma}

\noindent Lastly, if $T$ and $T'$ cannot be reduced any further under the subtree (resp. chain) reduction, we say that $T$ and $T'$ are {\it subtree} (resp. {\it chain}) {\it reduced}. \\

\noindent {\bf Agreement forests.}
Let $T$ and $T'$ be two unrooted binary phylogenetic trees  on $X$. Furthermore, let $F = \{B_0, B_1,B_2,\ldots,B_k\}$ be a partition of $X$, where each block $B_i$ with $i\in\{0,1,2,\ldots,k\}$ is  referred to as a \emph{component} of $F$. We say that $F$ is an \emph{agreement forest} for $T$ and $T'$ if the following hold. 
\begin{enumerate}
\item For each $i\in\{0,1,2,\ldots,k\}$, we have $T|B_i = T'|B_i$.
\item For each pair $i,j\in\{0,1,2,\ldots,k\}$ with $i \neq j$, we have that
$T[B_i]$ and $T[B_j]$ are vertex-disjoint in $T$, and $T'[B_i]$ and $T'[B_j]$ are vertex-disjoint in $T'$. 
\end{enumerate}

Let $F=\{B_0,B_1,B_2,\ldots,B_k\}$ be an agreement forest for $T$ and $T'$.
The \emph{size} of $F$ is simply \fuchsia{its} number of components; i.e. $k+1$. Moreover, an agreement forest with the minimum number of components (over all agreement forests for $T$ and $T'$) is called a \emph{maximum agreement forest} for $T$ and $T'$. The number of components of a maximum agreement forest for $T$ and $T'$ is denoted by $d_\MAF(T,T')$. The following theorem is well known.

\begin{theorem}\cite[Theorem 2.13]{AllenSteel2001}
Let $T$ and $T'$ be two unrooted binary phylogenetic trees  on $X$. Then $$d_\TBR(T,T') = d_\MAF(T,T') - 1.$$
\end{theorem}

Again, let $F$ be an agreement forest  for two unrooted binary phylogenetic trees $T$ and $T'$ on $X$, and let $C=(\ell_1,\ell_2,\ldots,\ell_n)$ be a common $n$-chain of $T$ and $T'$. We say that $C$ is \emph{split} in $F$ if there exist (at least) two components, say $B_j$ and $B_{j'}$, with $j\ne j'$ such that $B_j\cap C\ne\emptyset$ and $B_{j'}\cap C\ne\emptyset$.  Furthermore, we say that $C$ is \emph{atomized} in $F$ if each taxon $\ell_i \in C$ with $i\in\{1,2,\ldots,n\}$ occurs as a singleton component $\{\ell_i\}$ in $F$. Lastly, we say that $C$ is \emph{preserved} in $F$ if there exists a component $B_j \in F$ such that $C \subseteq B_j$. \green{Taking} the last three definitions together we have that $C$ is split in $F$ if and only if it is not preserved in $F$.\\




\section{A new suite of reduction rules}\label{sec:new-red}

In 2001, Allen and Steel~\cite{AllenSteel2001} showed that computing the TBR distance between two unrooted binary phylogenetic trees $T$ and $T'$ is fixed-parameter tractable, when parameterized by  $k=d_\TBR(T,T')$. To this end, they established a linear kernel of size \fuchsia{at most} $28k$. Recently, this result was improved by Kelk and Linz~\cite{tightkernel} who showed \green{with a new analysis} that the following \green{superior bound actually holds}.

\begin{theorem}\label{t:tbr-kernelEARLY}
Let $T$ and $T'$ be two subtree and chain reduced unrooted binary phylogenetic trees on $X$. If $d_\TBR(T,T')\geq 2$, then $|X|\leq 15 d_\TBR(T,T')-9$. 
\end{theorem}

\noindent Noting that the subtree and chain reduction can be applied in time that is polynomial in the size of $X$, it immediately follows that computing the TBR distance is fixed-parameter tractable when parameterized by $k$. In what follows, we will develop five new reduction rules that complement the subtree and chain reduction as introduced by Allen and Steel~\cite{AllenSteel2001} and further \green{improve the
$15\cdot d_\TBR(T,T') -9$ bound.}\\

Although not directly addressed by \cite{AllenSteel2001}, reducing a chain of length 3 to length 2 can strictly lower the TBR distance, which is why their chain reduction only allows reductions to length 3. An explicit example of this phenomenon are the two phylogenetic trees that are shown in Figure~\ref{fig:trees}. The TBR distance of these two trees is 2. However, if we reduce the common
3-chain $(b,c,d)$ to the 2-chain $(b,c)$, we obtain the two quartets $ab|ce$ and $eb|ca$ whose TBR distance is 1. Nevertheless, as we shall see, there \emph{do} exist  special circumstances when common 3-chains can be further reduced without altering the TBR distance. This is an important building block of our new reduction rules which ultimately will yield a kernel of size at most $11\cdot d_\TBR(T,T')-9$. To obtain this bound, we  combine the generator approach introduced in \cite{tightkernel} with a careful analysis of agreement forests. \\

The next theorem, which is formally established in the appendix to avoid disrupting the flow of the paper, will  repeatedly be used in establishing our new kernel result. 

\begin{theorem}
\label{thm:allchainsintact}
\green{
Let $T$ and $T'$ be two unrooted binary phylogenetic trees on $X$. 
Let $K$ be an (arbitrary) set of mutually taxa-disjoint chains that are common to $T$ and $T'$.
Then there exists a maximum agreement forest $F'$ of $T$ and $T'$ such that }
\begin{enumerate}
\item \green{every $n$-chain in $K$ with $n\geq 3$
is preserved in $F'$, and}
\item \green{every 2-chain in $K$ 
that is pendant in at least one of $T$ and $T'$
is preserved in $F'$.}
\end{enumerate}
\end{theorem}
\green{Theorem~\ref{thm:allchainsintact} is somewhat more general than
we need in this article - for us, $|K|\leq 2$ is sufficient - but we include full details because we consider it to be of independent interest and anticipate future applications beyond this article.}
\fuchsia{Furthermore, we} remark that a weaker version of Theorem~\ref{thm:allchainsintact} was already presented in \cite{AllenSteel2001}, where it forms the foundation of the proof of Lemma~\ref{t:safe}. However, the authors of \cite{AllenSteel2001} did not prove any results about chains of length 2, and their proof mainly \blueblue{focuses} on the case of a \green{\emph{single}} common chain of length 3 \green{that is pendant in \emph{neither} of the} two trees. \\



%
%

Throughout the next four \fuchsia{subsections}, we detail five new \fuchsia{reduction rules}. We will see that, similar to Lemma~\ref{t:safe}, each of these new reductions either preserves the TBR distance or reduces the parameter by 1. To simplify the exposition, we assume that the new reductions are always applied to two unrooted binary phylogenetic trees that are subtree and chain reduced. Lastly, while the reduction names appear to be cryptic, they will be further explained in Section~\ref{sec:new-kernel}, where we tie the new reductions and a careful unrooted generator analysis together to establish an improved kernel result. A generic example for each of the five new reductions is shown in Figure~\ref{fig:reductions}.

\subsection{$(*,3,*)$-reduction}
\label{red:reduc1}

Let $T$ and $T'$ be two unrooted binary phylogenetic trees on $X$, and let $C = (a,b,c)$ be a 3-chain that is common to $T$ and $T'$. For example $C$ may be the result of a previously applied chain reduction. If  $T$ has cherry $\{b,c\}$ and $T'$ has has cherry $\{a,b\}$, then a {\it $(*,3,*)$-reduction} is the operation of deleting $a$, $b$, and $c$ from $T$ and $T'$. Formally, we set $T_r = T|X \setminus C$ and $T'_r = T'|X \setminus C$.

\begin{lemma}\label{l:reduc1}
Let $T$ and $T'$ be two unrooted binary phylogenetic trees on $X$, and let $T_r$ and $T_r'$ be two trees obtained from $T$ and $T'$, respectively, by a single application of the $(*,3,*)$-reduction. Then $d_\TBR(T_r, T'_r) = d_\TBR(T,T') - 1$.
\end{lemma}

\begin{proof}
Without loss of generality, we establish the lemma using the same notation as in the definition of a $(*,3,*)$-reduction. \fuchsia{Let} $F_r$ be a maximum agreement forest for $T_r$ and $T_r'$, \green{and let $F$ be a maximum agreement forest for $T$ and $T'$.} Then $F_r\cup \{\{a,b,c\}\}$ is an agreement forest for $T$ and $T'$ \green{which implies that $|F_r| \geq |F|-1$. Hence,}
$$\green{d_\TBR(T_r, T'_r) = |F_r|-1 \geq |F|-2 = d_\TBR(T,T') - 1.}$$
By Theorem \ref{thm:allchainsintact}, we may assume that $C$ is preserved in $F$.
\green{(Formally, we apply the theorem to the set $K=\{ C \}$.)}
Hence, there exists a component $B_{abc}$ in $F$ such that  $C\subseteq B_{abc}$. Towards a contradiction, assume that $C\subset B_{abc}$ and let $x\in B_{abc} \setminus C$. Then, as $\{b,c\}$ is a cherry in $T$ and $\{a,b\}$ is a cherry of $T'$, it follows that $bc|ax$ is a quartet of $T|B_{abc}$ and $ab|cx$ is a quartet of $T'|B_{abc}$; thereby contradicting that $F$ is an agreement forest for $T$ and $T'$. Now, since $B_{abc}=C$, we have that 
\green{$F\setminus \{B_{abc}\}$}
 is an agreement forest for $T_r$ and $T_r'$. \green{Hence, $|F_r| \leq |F|-1$, which yields}
\green{$$d_\TBR(T, T') -1= |F|-2 \geq |F_r|-1 = d_\TBR(T_r,T_r').$$}\qed
\end{proof}

\begin{figure}[h!]
\center
\scalebox{1}{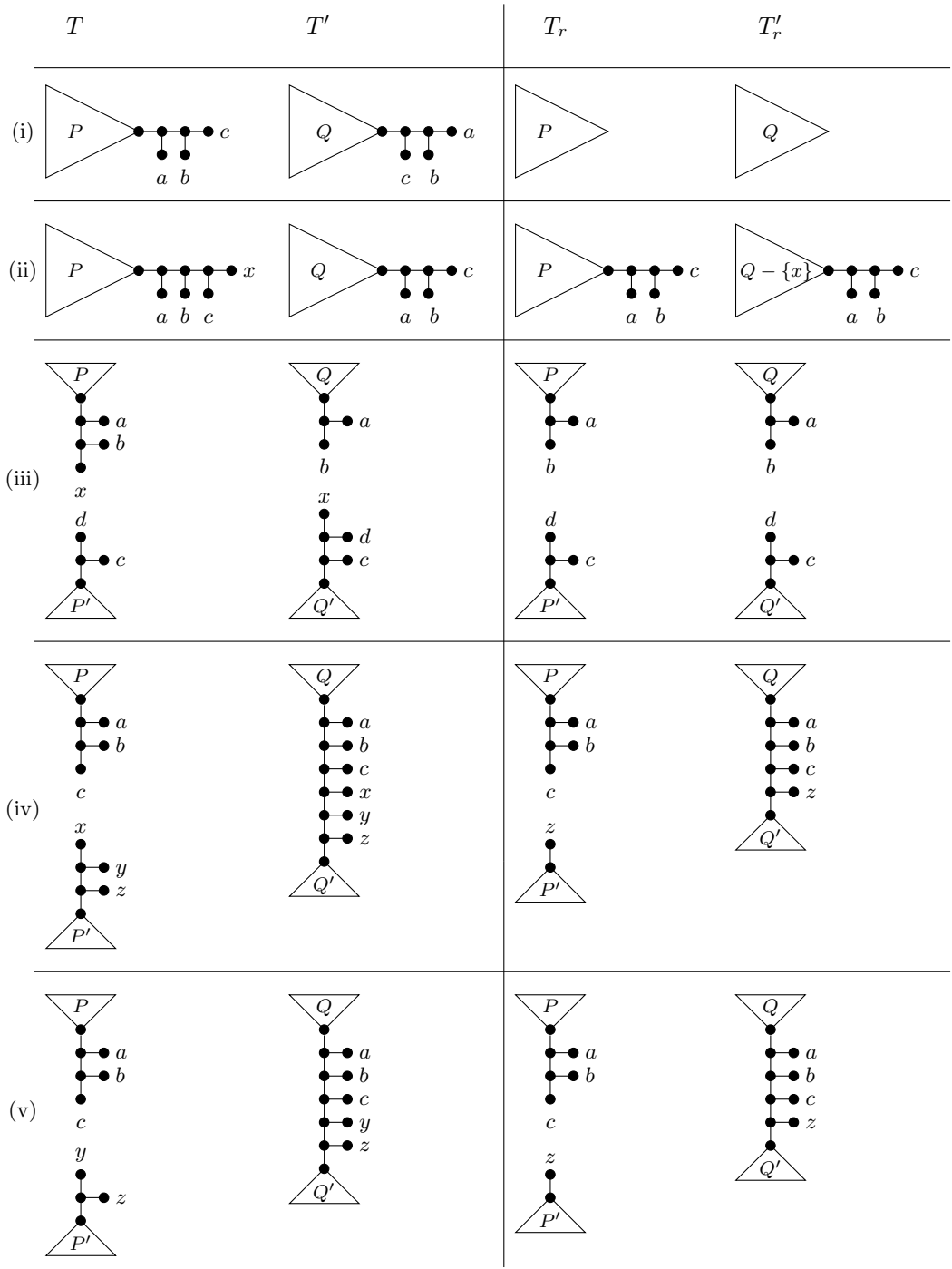}
\caption{The five reductions that are described in Subsections~\ref{red:reduc1}-\ref{red:reduc4}: (i) $(*,3,*)$-reduction, (ii) $(3,1,*)$-reduction, (iii) $(2,1,2)$-reduction, (iv) $(3,3)$-reduction, (v) $(3,2)$-reduction. Triangles indicate subtrees. For (iii)-(v), we have omitted some parts of the trees. \green{Note that the reductions do not require the
sets $P$, $Q$, $P'$, and $Q'$ to all be non-empty.}}
\label{fig:reductions}
\end{figure}

\subsection{$(3,1,*)$-reduction}
\label{red:reduc2}

Let $T$ and $T'$ be two unrooted binary phylogenetic trees on $X$, and let $C = (a,b,c)$ be a 3-chain that is common to $T$ and $T'$. If  $T'$  has cherry $\{b,c\}$ and $T$ has cherry $\{c,x\}$ for some element $x\in X\setminus\fuchsia{C}$, then a {\it $(3,1,*)$-reduction} is the operation of deleting $x$ from $T$ and $T'$, i.e. we set $T_r = T|X \setminus \{x\}$ and $T'_r = T'|X \setminus \{x\}$. Informally, $x$ prevents $C$ from being a common pendant subtree of $T$ and $T'$.

\begin{lemma}\label{l:reduc2}
Let $T$ and $T'$ be two unrooted binary phylogenetic trees on $X$, and let $T_r$ and $T_r'$ be two trees obtained from $T$ and $T'$, respectively, by a single application of the $(3,1,*)$-reduction. Then
$d_\TBR(T_r, T'_r) = d_\TBR(T,T') - 1.$
\end{lemma}

\begin{proof}
Again without loss of generality, we establish the lemma using the same notation as in the definition of a $(3,1,*)$-reduction. \fuchsia{Let $F_r$ be a maximum agreement forest for $T_r$ and $T_r'$, and let $F$ be a maximum agreement forest for $T$ and $T'$.} To show that $d_\TBR(T_r, T'_r) \geq d_\TBR(T,T') - 1,$ we apply the same argument as in the first part of the proof of Lemma~\ref{l:reduc1} with the modification of considering $F_r\cup\{\{x\}\}$ (instead of $F_r\cup\{C\}$) as an agreement forest for $T$ and $T'$. \fuchsia{Now, by} Theorem~\ref{thm:allchainsintact}, we may assume that $C$ is preserved in $F$. Let $B_x$ be the component of $F$ that contains $x$, and let $B_{abc}$ be the component of $F$ such that $C\subseteq B_{abc}$. By the choice of $F$, $B_{abc}$ exists. Clearly, $B_x\ne B_{abc}$ since, otherwise, $ab|cx$ is a quartet of $T|B_x$ but not a quartet of $T'|B_x$. Moreover, if $|B_x|\geq 2$, then $T[B_x]$ and $T[B_{abc}]$ are not vertex-disjoint in $T$. It now follows that $B_x=\{x\}$ and that $F\setminus \{B_x\}$ is an agreement forest for  $T_r$ and $T_r'$. \green{Hence, $|F_r| \leq |F|-1$, so we have,}
\green{$$d_\TBR(T, T') -1= |F|-2 \geq |F_r|-1 = d_\TBR(T_r,T_r').$$} \qed
\end{proof}

\noindent Following on from the proof of Lemma~\ref{l:reduc2}, note that $\{a,b,c\}$ is the leaf set of a pendant subtree that is common to $T_r$ and $T_r'$. The two reduced trees can therefore be further reduced by the subtree reduction.  

\subsection{$(2,1,2)$-reduction}
\label{red:reduc3}
Let $T$ and $T'$ be two unrooted binary phylogenetic trees on $X$. Furthermore, let $\{a,b,c,d,x\}\subseteq X$ such that $C_1=(a,b)$ and $C_2=(c,d)$ are two 2-chains that are common to $T$ and $T'$. 
If $T$ has cherries $\{b,x\}$ and $\{c,d\}$, and if $T'$ has cherries $\{a,b\}$ and $\{d,x\}$, then a {\it $(2,1,2)$-reduction} is the operation of obtaining $T_r$ and $T_r'$ from $T$ and $T'$, respectively, by deleting $x$ from $T$ and $T'$, i.e. $T_r = T|X \setminus \{x\}$ and $T'_r = T'|X \setminus \{x\}$.

\begin{lemma}\label{l:reduc3}
Let $T$ and $T'$ be two unrooted binary phylogenetic trees on $X$, and let $T_r$ and $T_r'$ be  two trees obtained from $T$ and $T'$, respectively, by a single application of the $(2,1,2)$-reduction. Then
$d_\TBR(T_r, T'_r) = d_\TBR(T,T') - 1.$
\end{lemma}

\begin{proof}
Without loss of generality, we may assume that the two common 2-chains $C_1$ and $C_2$ and their respective
configurations
in $T$ and $T'$ are exactly as described in the definition of a $(2,1,2)$-reduction. Then, $d_\TBR(T_r, T'_r) \geq d_\TBR(T,T') - 1$  follows as described in the proof of Lemma~\ref{l:reduc2}. We establish the lemma by showing that $d_\TBR(T_r, T'_r) \leq d_\TBR(T,T') - 1$. Let $F$ be a maximum agreement forest for $T$ and $T'$, \fuchsia{and let $F_r$ be a maximum agreement forest for $T_r$ and $T_r'$.} By Theorem~\ref{thm:allchainsintact}, we may assume that $C_1$ and $C_2$ are preserved in $F$. \green{(Formally, we apply the theorem to the set of chains $K=\{ C_1, C_2 \}$, noting that each chain is pendant in one of the two trees.)} Let $B_{ab}$ be the element in $F$ that contains $C_1$ and, similarly,  let $B_{cd}$ be the element in $F$ that contains $C_2$. Towards showing that $\{x\}\in F$, first assume that there exists an element $B_x\in F$ such that $|B_x|\geq 2$ and $B_x\cap\{a,b,c,d,x\}=\{x\}$. Then, it is straightforward to check that $T[B_x]$ and $T[B_{ab}]$ are not vertex-disjoint in $T$; a contradiction. Thus, $x$ is either contained in $B_{ab}$ or $B_{cd}$, or $\{x\}$ is an element in $F$. Now, if $B_{ab}=B_{cd}$ and $x\in B_{ab}$, then $ax|cd$ is a quartet of $T|B_{ab}$ while $ac|dx$ is a quartet of $T'|B_{ab}$; a contradiction. Otherwise, if $B_{ab}\ne B_{cd}$ and $x\in B_{ab}$, then $T'[B_{ab}]$ and $T'[B_{cd}]$ are not vertex-disjoint in $T'$. Symmetrically, if $B_{ab}\ne B_{cd}$ and $x\in B_{cd}$, then $T[B_{ab}]$ and $T[B_{cd}]$ are not vertex-disjoint in $T$. It now follows that $\{x\}\in F$ and that $F\setminus \{\{x\}\}$ is an agreement forest for  $T_r$ and $T_r'$. Hence, we have
\green{$$d_\TBR(T, T') -1= |F|-2 \geq |F_r|-1= d_\TBR(T_r,T_r').$$} 
\qed
\end{proof}
After performing a $(2,1,2)$-reduction, note that the two reduced trees $T_r$ and $T_r'$ have common pendant subtrees on leaf sets $\{a,b\}$ and $\{c,d\}$, respectively, that can be reduced further under the subtree reduction.

\subsection{$(3,3)$- and $(3,2)$-reduction}
\label{red:reduc4}
Let $T$ and $T'$ be two unrooted binary phylogenetic trees on $X$. The next reduction can be applied in two slightly different situations. The first situation considers two 3-chains while the second situation considers one 3-chain and one 2-chain. We start by formally describing the first situation. Let $\{a,b,c,x,y,z\}\subseteq X$ such that $C_1=(a,b,c)$ and $C_2=(x,y,z)$ are two 3-chains that are common to $T$ and $T'$. If $T$ has cherries $\{b,c\}$ and $\{x,y\}$, and if $T'$ has a 6-chain $(a,b,c,x,y,z)$ 
then a {\it $(3,3)$-reduction} is the operation of obtaining $T_r$ and $T_r'$ from $T$ and $T'$, respectively, by deleting $x$ and $y$ from $T$ and $T'$, i.e. $T_r = T|X \setminus \{x,y\}$ and $T'_r = T'|X \setminus \{x,y\}$. 
 We now turn to the second situation. Let $\{a,b,c,y,z\}\subseteq X$ such that $C_1=(a,b,c)$ and $C_2=(y,z)$ are two chains that are common to $T$ and $T'$. If $T$ has cherries $\{b,c\}$ and $\{y,z\}$, and if $T'$ has a 5-chain $(a,b,c,y,z)$, then a {\it $(3,2)$-reduction} is the operation of obtaining $T_r$ and $T_r'$ from $T$ and $T'$, respectively by deleting $y$ from $T$ and $T'$, i.e. $T_r = T|X \setminus \{y\}$ and $T'_r = T'|X \setminus \{y\}$.

\begin{lemma}\label{l:reduc4}
Let $T$ and $T'$ be two unrooted binary phylogenetic trees on $X$, and let $T_r$ and $T_r'$ be two trees obtained from $T$ and $T'$, respectively, by a single application of the $(3,3)$- or the $(3,2)$-reduction. Then $d_\TBR(T_r, T'_r) = d_\TBR(T,T').$
\end{lemma}

\begin{proof}
Without loss of generality, we may assume that the two common chains $C_1$ and $C_2$ and their respective configurations in $T$ and $T'$ are exactly as described in the paragraph prior to the statement of this lemma. Let $Y=\{a,b,c,x,y,z\}$ if $T_r$ and $T_r'$ have been obtained from $T$ and $T'$ by a $(3,3)$-reduction and, otherwise, let $Y=\{a,b,c,y,z\}$. \\

%
%
\green{Given that the $T_r$ and $T'_r$ are induced subtrees of $T$ and $T'$ respectively (i.e. obtained from $T$ and $T'$ using the ``$|$'' operator), it follows from Lemma 2.11 of \cite{AllenSteel2001} that} $$\green{d_\TBR(T_r, T'_r)\leq d_\TBR(T,T').}$$ To establish the other direction, let $F_r$ be a maximum agreement forest for $T_r$ and $T_r'$, \fuchsia{and let $F$ be a maximum agreement forest for $T$ and $T'$.} By Theorem~\ref{thm:allchainsintact}, we may assume that $C_1$ is preserved in $F_r$. Let $B_{abc}$ be the element in $F_r$ such that $C_1\subseteq B_{abc}$. Similarly, let $B_z$ be the element in $F_r$ such that $z\in B_z$. We have $B_{abc}\ne B_z$ since, otherwise, $bc|az$ is a quartet of $T_r|B_{abc}$ while $ab|cz$ is a quartet of  $T_r'|B_{abc}$; a contradiction. \green{Next, observe that if $B_{abc}$ contains some taxon $d \not \in \{a,b,c,z\}$, then $d$ \fuchsia{is a leaf in the subtree $Q$  of $T_r'$}, as depicted in
Figure \ref{fig:reductions}(iv) and (v). If this was not so, then $ad|bc$ would be a quartet of $T_r|B_{abc}$ while $cd|ab$ would be a quartet of $T_r'|B_{abc}$; a contradiction. Combining these facts yields the insight that, in $T_r'$, the edge between the parent of $z$ and the parent of $c$ (if such an edge exists) is not on the embedding of any component in $F_r$.} Since $F_r$ is an agreement forest for $T_r$ and $T_r'$, it now follows that $$(F_r\setminus\{B_z\})\cup\{B_z\cup\{x,y\}\}$$ is an agreement forest for $T$ and $T'$ if a $(3,3)$-reduction has been applied and that $$(F_r\setminus\{B_z\})\cup\{B_z\cup\{y\}\}$$ is an agreement forest for $T$ and $T'$ if a $(3,2)$-reduction has been applied. Hence, 

\green{$$d_\TBR(T_r, T'_r) = |F_r|-1 \geq |F|-1 = d_\TBR(T,T').$$}
\qed
\end{proof}


We end this section by noting that it takes \green{$O( \text{poly}(|X|))$} time to test if any of the new reductions presented in Subsections~\ref{red:reduc1}-\ref{red:reduc4} can be applied. While the $(3,2)$- and $(3,3)$-reduction preserves the TBR distance, each of the other three new reductions reduces the parameter by exactly one, i.e. the TBR distance for the unreduced trees can be calculated by computing the TBR distance for the reduced trees and adding one to the result.

\section{A new kernel for computing the TBR distance}\label{sec:new-kernel}

In this section, we establish the main result of this paper. To formally state it, we require a new definition. Let $T$ and $T'$ be two unrooted binary phylogenetic trees on $X$. We say that $T$ and $T'$ are {\it exhaustively reduced} if they are subtree and chain reduced, and none of the five reductions presented in Section~\ref{sec:new-red} can be applied to $T$ and $T'$.

\begin{theorem}\label{t:11kernel}
Let $T$ and $T'$ be two exhaustively reduced unrooted binary phylogenetic trees on $X$. If $d_\TBR(T,T')\geq 2$, then $|X|\leq 11d_\TBR(T,T')-9$.
\end{theorem}

\noindent To establish this theorem, we analyze the maximum size of two exhaustively reduced phylogenetic trees with the help of \blueblue{an unrooted binary phylogenetic network $N$ that displays the two trees and the unrooted  generator that underlies $N$. Next, we define unrooted generators.} \\

Let $k$ be a positive integer.
For $k\geq 2$, a {\it $k$-generator} (or short {\it generator} when $k$ is clear from the context) is a connected cubic multigraph with edge set $E$ and vertex set $V$ such that $k=|E|-(|V|-1)$. 
The edges of a generator are  called its {\it sides}. Intuitively, given an unrooted binary phylogenetic network $N$ with $r(N)=k$,
we can obtain a $k$-generator by, repeatedly, deleting all (labeled and unlabeled) leaves and suppressing any resulting degree-2 vertices. We say that \green{the generator obtained \fuchsia{in} this way} {\it underlies} $N$.  Now, let $G$ be a $k$-generator, let $\{u,v\}$ be a side of $G$, and let $Y$ be a set of leaves. The operation of subdividing $\{u,v\}$ with $|Y|$ new vertices and, for each such new vertex $w$, adding a new edge $\{w,\ell\}$, where $\ell\in Y$ and $Y$ bijectively labels the new leaves, is referred to as {\it attaching} $Y$  to  $\{u,v\}$. 
Lastly, if at least one new leaf is attached to each loop and to each pair of parallel edges in $G$, then the resulting graph is an unrooted binary phylogenetic network $N$ with $r(N)=k$. Note that $N$ has no pendant subtree with more than a single leaf. Hence, we have the following observation.

\begin{observation}\label{ob:gen}
Let $N$ be an unrooted binary phylogenetic network that has no pendant subtree with at least two leaves, and let $G$ be a generator. 
Then $G$ underlies $N$ if and only if $N$ can be obtained from $G$ by attaching a (possibly empty) set of leaves to each side of $G$.
\end{observation}

Now let $T$ and $T'$ be two subtree and chain reduced unrooted binary phylogenetic trees on $X$, and let $N$ be an unrooted binary phylogenetic network on $X$ that displays $T$ and $T'$.  \green{Let $S$ and $S'$ be spanning trees of $N$ obtained by greedily extending \fuchsia{a subdivision} of $T$ (respectively, $T'$) to become a spanning tree,
if it is not that already.}
Since $N$ displays $T$ and $T'$, $S$ and $S'$ exist. Furthermore, let $G$ be the generator that underlies $N$. Since $T$ and $T'$ are subtree and chain reduced, $N$ does not have a pendant subtree of size at least two. Hence, by Observation~\ref{ob:gen}, we can obtain $N$ from $G$ by attaching leaves to $G$. Let $s=\{u,w\}$ be a side of $G$. Let $Y=\{\ell_1,\ell_2,\ldots,\ell_m\}$ be the set of leaves that are  attached to $s$ in obtaining $N$ from $G$. Recall that $m \geq 0$. Then there exists a path $$u=v_0,v_1,v_2,\ldots,v_m,v_{m+1}=w$$ of vertices in $N$ such that, for each $i\in\{1,2,\ldots,m\}$, $v_i$ is the unique \fuchsia{parent} of $\ell_i$. We refer to this path as the {\it path associated with $s$} and denote it by $P_s$. Importantly, for a path $P_s$ in $N$ that is associated with a side $s$ of $G$, there is at most one edge in $P_s$ that is not contained in $S$, and there is at most one (not necessarily distinct) edge in $P_s$ that is not contained in $S'$. We  make this precise in the following definition and say that $s$ has {\it \green{$b$} breakpoints \fuchsia{relative} to $S$ and $S'$}, 
where 
\begin{enumerate}
\item \green{$b=0$} if $S$ and $S'$ both contain all edges of $P_s$,
\item \green{$b=1$} if one element in $\{S,S'\}$ contains all edges of $P_s$ while the other element contains all but one edge of $P_s$, and
\item \green{$b=2$} if each of $S$ and $S'$ contains all but one edge of $P_s$.
\end{enumerate}
Since $S$ and $S'$ span $N$, note that $s$ cannot have more than 2 breakpoints \fuchsia{relative} to $S$ and $S'$. \\

In the language of this paper, Kelk and Linz~\cite{tightkernel} have established the following result.

\begin{lemma}\label{lem:old-breakpoints}
Let $N$ be an unrooted binary phylogenetic network on $X$ that displays two subtree and chain reduced unrooted binary phylogenetic trees $T$ and $T'$. Let $S$ (resp. $S'$) be a spanning tree of $N$ \green{obtained by extending a} subdivision of $T$ (resp. $T'$). Furthermore, let $G$ be the generator that underlies $N$, and let $s$ be a side of $G$. Suppose that $s$ has \green{$b$} breakpoints \fuchsia{relative} to $S$ and $S'$ for some $\fuchsia{b}\in\{0,1,2\}$. Then,
\begin{enumerate}[(i)]
\item \fuchsia{if} \green{$b=0$}, then $N$ can be obtained from $G$ by attaching at most 3 leaves to $s$;
\item if \green{$b=1$}, then $N$ can be obtained from $G$ by attaching at most 6 leaves to $s$; and
\item if \green{$b=2$}, then $N$ can be obtained from $G$ by attaching at most 9 leaves to $s$. 
\end{enumerate}
\end{lemma}
Since Lemma~\ref{lem:old-breakpoints} only considers the subtree and chain reduction, a natural question is whether or not the  five reductions presented in Section~\ref{sec:new-red}  improve the bounds on the number of leaves that  are attached to a side of a generator. We answer this question positively  in the next lemma.

\begin{lemma}\label{lem:new-breakpoints}
Let $N$ be an unrooted binary phylogenetic network on $X$ that displays two exhaustively reduced unrooted binary phylogenetic trees $T$ and $T'$. Let $S$ (resp. $S'$) be a spanning tree of $N$ \green{obtained by extending a} subdivision of $T$ (resp. $T'$). Furthermore, let $G$ be the generator that underlies $N$, and let $s=\{u,v\}$ be a side of $G$. Suppose that $s$ has \green{$b$} breakpoints \fuchsia{relative} to $S$ and $S'$ for some $\fuchsia{b}\in\{0,1,2\}$. Then,
\begin{enumerate}[(i)]
\item \fuchsia{if} \green{$b=0$}, then $N$ can be obtained from $G$ by attaching at most 3 leaves to $s$;
\item if \green{$b=1$}, then $N$ can be obtained from $G$ by attaching at most 4 leaves to $s$; and
\item if \green{$b=2$}, then $N$ can be obtained from $G$ by attaching at most 4 leaves to $s$. 
\end{enumerate}
\end{lemma}

\begin{proof}
By Lemma~\ref{lem:old-breakpoints}(i), (i) follows immediately. \\

To establish (ii), we show that neither 5 nor 6 leaves are attached to $s$ and note that, by Lemma~\ref{lem:old-breakpoints}(ii), no more than 6 leaves are attached to $s$. Without loss of generality, we may assume that $S$ contains all edges of $P_s$ and that $S'$ contains all but one edge of $P_s$. Let $e$ be the edge of $P_s$ that $S'$ does not contain. First, assume that 6 leaves are attached to $s$. Let $P_s=v_0,v_1,v_2,\ldots,v_{6},v_7$.  Recall that $u=v_0$ and $v=v_7$.  For each $i\in\{1,2,\ldots,6\}$, let $\ell_i$ be the leaf adjacent to $v_i$ in $N$. If $e\ne\{v_3,v_4\}$, then $T$ and $T'$ have  a common chain of length at least 4; a contradiction since $T$ and $T'$ are chain reduced. On the other hand, if $e=\{v_3,v_4\}$, then $T$ and $T'$ have two common 3-chains $(\ell_1,\ell_2,\ell_3)$ and $(\ell_4,\ell_5,\ell_6)$ such that $(\ell_1,\ell_2,\ldots,\ell_6)$ is a chain of $T$, and  both of $\{\ell_2,\ell_3\}$ and $\{\ell_4,\ell_5\}$ are cherries of $T'$. Hence, $T$ and $T'$ can be further reduced under a $(3,3)$-reduction; again a contradiction. Second, assume that 5 leaves are attached to $s$. Let $P_s=v_0,v_1,v_2,\ldots,v_{5},v_6$. Since $T$ and $T'$ are chain reduced, we use an argument analogous to the previous 6-leaf case to show that $e\in\{\{v_2,v_3\},\{v_3,v_4\}\}$. If $e=\{v_2,v_3\}$, then $T$ and $T'$ have common chains $(\ell_1,\ell_2)$ and $(\ell_3,\ell_4,\ell_5)$ and \green{\fuchsia{$T$} has a chain $(\ell_1, \ell_2, \ell_3,\ell_4,\ell_5)$}, where $\ell_i$ is again the leaf adjacent to $v_i$ in $N$ for each $i\in\{1,2,\ldots,5\}$. Furthermore, \fuchsia{$T'$} has cherries \green{$\{\ell_1,\ell_2\}$ and $\{\ell_3,\ell_4\}$}. It now follows that  $T$ and $T'$ can be further reduced under a $(3,2)$-reduction; a contradiction to the fact that both trees are exhaustively reduced. If $e=\{v_3,v_4\}$, we use an \green{symmetric} argument to get the same contradiction; thereby establishing (ii). \\

We complete the proof by showing that (iii) holds. Throughout this part of the proof, suppose that at least 5 leaves get attached to $s$ in the process of obtaining $N$ from $G$ since, otherwise, (iii) follows without proof. 
Again, consider the path $$P_s=v_0,v_1,v_2,\ldots,v_{m+1}$$ that is associated with $s$ in $N$. Recall that $m$ is the number of leaves that are attached to $s$.	Hence $m\geq 5$. Let $\ell_i$ be the leaf adjacent to $v_i$ in $N$ for each $i\in\{1,2,\ldots,m\}$. Furthermore, let $e=\{v_i,v_j\}$  be the edge of $P_s$ that is not contained in $S$, and let $f=\{v_{i'},v_{j'}\}$ be the edge of $P_s$ that is not contained in $S'$. Without loss of generality, we may assume that
\green{$i=j-1$, $i' = j'-1$,}
and $i\leq i'$. Moreover, note that \green{if $i < i'$ then} $C=(\ell_{i+1},\ell_{i+2},\ldots,\ell_{i'})$ is an $(i'-i)$-chain that is common to $T$ and $T'$. Considering four cases and deriving a contradiction for each, we next show that $i'-i=1$. 
\begin{description}
\item [Case 1.] If $i'-i>3$, then $C$ has length at least 4 and $T$ and $T'$ are not chain reduced. 
\item [Case 2.] If $i'-i=3$, then $C=(\ell_{i+1},\ell_{i+2},\ell_{i+3})$ is a maximal common 3-chain of $T$ and $T'$. Moreover, as $\{v_{i+1},v_{i+2}\}$ is a cherry of $T$ and $\{v_{i+2},v_{i+3}\}$ is a cherry of $T'$, it  follows that a $(*,3,*)$-reduction can be applied to $T$ and $T'$; thereby contradicting that $T$ and $T'$ are exhaustively reduced.
\item [Case 3.] If $i'-i=2$, then $C=(\ell_{i+1},\ell_{i+2})$ is a maximal common 2-chain of $T$ and $T'$. In particular \fuchsia{$C$} is the leaf set of a pendant subtree that is common to $T$ and $T'$ that can be further reduced under the subtree reduction.
\item [Case 4.] If $i'-i=0$, then $\{\ell_1,\ell_2,\ldots,\ell_i\}$ and $\{\ell_j,\ell_{j+1},\ldots,\ell_{m-1}\}$ are the leaf sets of two pendant subtrees that are common to $T$ and $T'$. Since $m\geq 5$ one of these subtrees has size at least two and, so, $T$ and $T'$ can be further reduced under the subtree reduction.
\end{description}
All four cases contradict the fact that $T$ and $T'$ are exhaustively reduced. Thus, we may assume for the remainder of the proof that, if $m\geq 5$, then  $i'-i=1$.\\

We next establish a maximum for $i$ and  minimum for $i'$. Clearly, if $i>3$, then $(\ell_1,\ell_2,\ldots,\ell_i)$ is a chain of length at least 4 that is common to $T$ and $T'$ that can be reduced  by applying a chain reduction. Moreover, if $i=3$, first recall that $i'=i+1$. It then follows that $(\ell_1,\ell_2,\ell_3)$ is a chain that is common to $T$ and $T'$, $\{\ell_2,\ell_3\}$ is a cherry of $T$, and $\{\ell_3,\ell_4\}$ is a cherry of $T'$. Hence, $T$ and $T'$ can be further reduced by applying a $(3,1,*)$-reduction, where $\ell_4$ takes on the role of $x$ in the definition of this reduction. Hence $i\leq 2$. By symmetry and applying an analogous argument, we derive that $i'\geq m-2$. In summary, under the assumption that $m\geq 5$, we have established the following three restrictions on the indices $i$ and $i'$:
$$i\leq 2;$$
$$i'= i+1, \text{ and}$$
$$i'\geq m-2.$$
Taken all three together, it follows that $m\leq 5$. So suppose that $m=5$. Then, \fuchsia{by the aforementioned three restrictions}, this implies that $e=\{v_2,v_3\}$ and $f=\{v_3,v_4\}$. Furthermore, $(\ell_1,\ell_2)$ and $(\ell_4,\ell_5)$ are two  2-chains that are common to $T$ and $T'$ such that $T$ has cherries $\{\ell_1,\ell_2\}$ and $\{\ell_3,\ell_4\}$, and $T'$ has cherries $\{\ell_2,\ell_3\}$ and $\{\ell_4,\ell_5\}$. With $\ell_3$ taking on the role of $x$ in the definition of a $(2,1,2)$-reduction, it now follows that $T$ and $T'$ can be further reduced under this reduction.  This contradicts our initial assumption that $m\geq 5$; thereby establishing (iii).
\qed
\end{proof}

\green{\fuchsia{We can now} clarify the rather cryptic names of the new
reduction rules. From the proof of Lemma \ref{lem:new-breakpoints} we can see that a side $s$ that has \fuchsia{$2$} breakpoints, indexed by $i$ and $i'$ respectively (where $i < i'$), induces three common chains of length $i$, $i'-i$ and $m-i'$. We can summarize these three lengths in a vector $(i, i'-i, m-i')$. \fuchsia{Then the $(*,3,*)$-reduction} can be applied
when $i'-i = 3$, irrespective of the values of $i$ and $m-i'$, and we denote this indifference using wildcard symbols. The same idea applies to the \fuchsia{$(3,1,*)$- and the $(2,1,2)$-reduction}. For sides with a single breakpoint at position $i$, the vector of common chain lengths induced is given by $(i, m-i)$. \fuchsia{Then essentially, the $(3,3)$- and the $(3,2)$-reduction} capture the situation when a 6-chain (resp. 5-chain) in $T'$ is split into two shorter chains in $T$ by a breakpoint at $i=3$.}
\\

We are now in a position to establish Theorem~\ref{t:11kernel}.\\

\noindent {\it Proof of Theorem~\ref{t:11kernel}.}
Let $N$ be an unrooted binary phylogenetic network on $X$ that displays $T$ and $T'$ such that $$r(N)=h^u(T,T')=d_\TBR(T,T')=k\geq 2,$$ where the second equality follows from Theorem~\ref{t:tbr-equiv}. Let $S$ and $S'$ be spanning trees of $N$ that are \green{obtained by extending} subdivisions of $T$ and $T'$, respectively. Furthermore, let $G$ be the generator that underlies $N$. To establish the theorem, we use Lemma~\ref{lem:new-breakpoints} to bound from above the number of leaves that can collectively be attached to $G$ over all of its sides. The following approach is similar to the one used in~\cite[Lemma 3]{tightkernel}. 
\green{By ~\cite[\fuchsia{Lemma 1}]{tightkernel}, $G$ has $3(k-1)$ sides. \fuchsia{Furthermore} $N$ contains exactly $k$ edges that are not contained in $S$, and exactly $k$ edges that are not contained in $S'$. Each of these edges induces a breakpoint on a corresponding side of $G$, so each side of $G$ can have 0, 1 or 2 breakpoints and there are $2k$ breakpoints in total. Let $q$ be the number of sides in $G$ that have two breakpoints relative to $S$ and $S'$. Noting that $0\leq q\leq k$, it follows that there are $2(k-q)$ sides in $G$ whose number of breakpoints is one relative to $S$ and $S'$.} Hence, there are $3(k-1)-(2k-q)$ sides in $G$ that have zero breakpoints relative to $S$ and $S'$. Since $T$ and $T'$ are exhaustively reduced, we now apply Lemma~\ref{lem:new-breakpoints} to derive the following inequality: $$|X|\leq 4q+4(2(k-q))+3(3(k-1)-(2k-q))=-q+11k-9.$$ Clearly, $-q+11k-9$ is maximum for $q=0$ and, so, we have $$|X|\leq -q+11k-9\leq 11k-9=11d_\TBR(T,T')-9$$ which establishes the theorem.
\qed

\mbox{}\\

We finish the section with an additional kernel result that establishes an even smaller kernel for particular trees.\\

\begin{corollary}
Let $T$ and $T'$ be two unrooted binary phylogenetic trees on $X$. If $d_\TBR(T,T')\geq 2$,  $T$ and $T'$ are subtree reduced and do not have any common $n$-chain with $n\geq 2$, then  $|X|\leq 5d_\TBR(T,T') - 3$.
\end{corollary}
\begin{proof}
As previously, let $N$ be an unrooted binary phylogenetic network on $X$ that displays $T$ and $T'$ and has \fuchsia{the property that} $r(N)=h^u(T,T')$. Let $S$ and $S'$ be spanning trees of $N$ \green{obtained by extending} subdivisions of $T$ and $T'$, respectively, and let $G$ be the generator underlying $N$. For a side $s$ of $G$, it follows that we can attach at most 3 leaves to $s$ if $s$ has two breakpoints relative to $S$ and $S'$. Similarly, we can attach at most 2 leaves (resp. 1 leaf) to $s$ if $s$ has \green{one} breakpoint (resp. \green{zero} breakpoints) relative to $S$ and $S'$. Interestingly, and in comparison with the proof of Lemma~\ref{lem:new-breakpoints}, these upper bounds can be easily established using arguments that only rely on the (ordinary) subtree and chain reduction, but make no use of the five reductions presented in Section~\ref{sec:new-red}. Now, using the same counting argument as in the proof of Theorem~\ref{t:11kernel}, we have $$|X|\leq 5k-3=5d_\TBR(T,T')-3.$$\qed
 \end{proof}

\begin{figure}[h!]
\center
\scalebox{1}{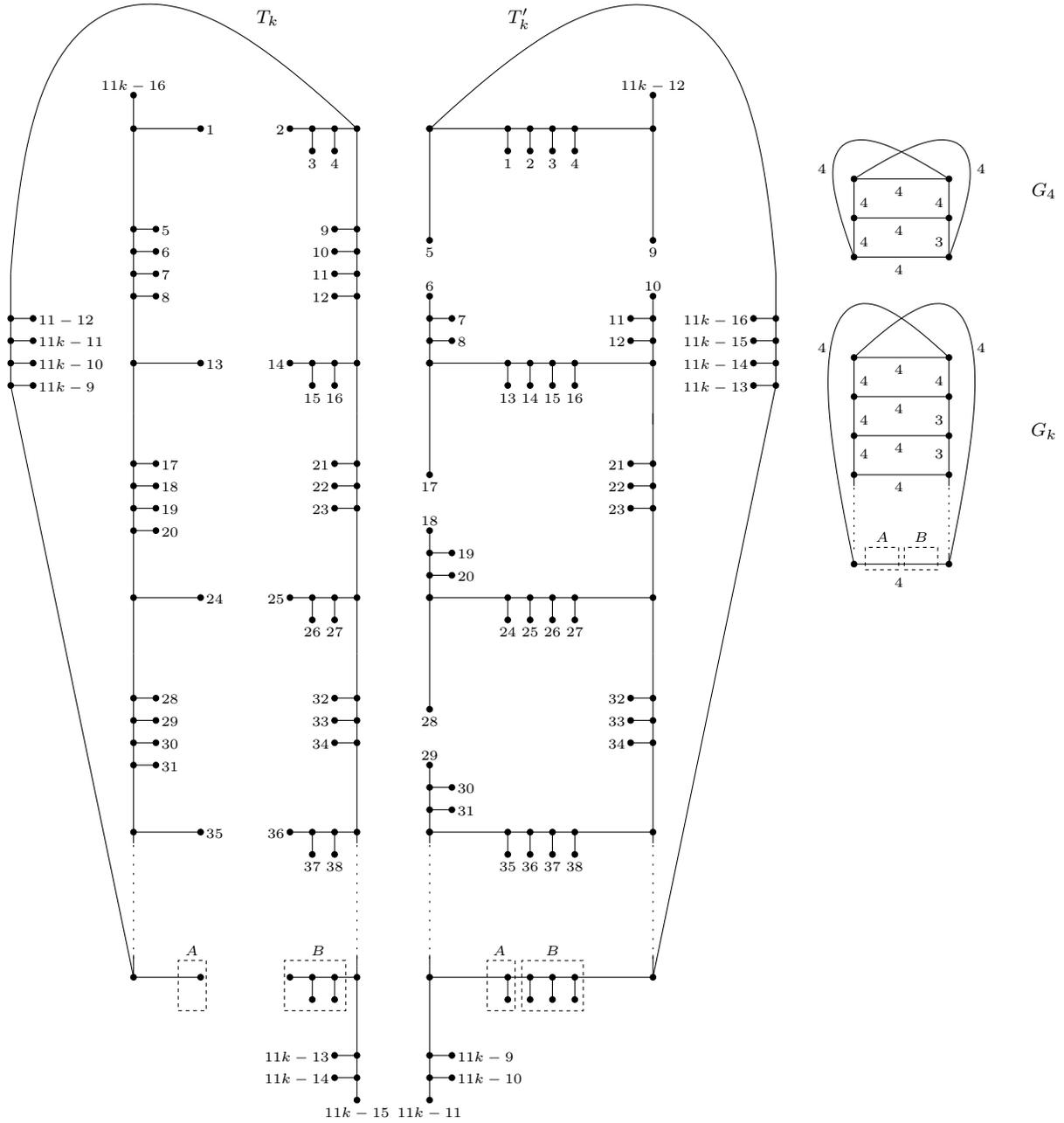}
\caption{Two exhaustively reduced unrooted binary phylogenetic trees $T_k$ and $T'_k$ as well as the generator $G_k$ (and $G_4$) that provide a family of trees to show that the kernel presented in Theorem~\ref{t:11kernel} is tight for each $k=d_\TBR(T_k,T_k')\geq 4$. Blocks $A$ and $B$ indicate a (common) 1-chain and 3-chain, respectively. For details see \green{the main} text and~\cite[Section 4]{tightkernel}. }
\label{fig:tight}
\end{figure}

\section{Tightness of the kernel under the new reductions}
In this section, we show that, for two exhaustively reduced trees, the kernel result presented in Theorem~\ref{t:11kernel} is tight. For each $k\geq 4$, we do this by providing two \green{exhaustively reduced} unrooted binary phylogenetic trees $T_k$ and $T_k'$ whose leaf sets have size $11k-9$, such that  \green{$d_\TBR(T_k,T_k')=k$}. To illustrate, $T_k$ and $T_k'$ are shown in Figure~\ref{fig:tight}. It is straightforward to check that $T_k$ and $T_k'$ are exhaustively reduced. While we do not go into detail about justifying that $T_k$ and $T_k'$ indeed provide a tight example, i.e. $d_\TBR(T_k,T_k')=k$, we point the interested reader to~\cite[Section 4]{tightkernel}, where a very similar family of constructions is given to show that the kernel result presented in~\cite{tightkernel} is tight for phylogenetic trees that are subtree and chain reduced, and do not contain any common so-called cluster. The approach taken there, which uses unrooted generators to argue that 
$d_\TBR(T_k,T_k') \leq k$ and \emph{maximum parsimony distance} \cite{fischer2014} to prove $d_\TBR(T_k,T_k') \geq k$, can be easily adapted to establish the following proposition from which tightness of the kernel presented in Theorem~\ref{t:11kernel} immediately follows\footnote{\green{In fact, \fuchsia{up to relabeling of the leaves} the trees shown here are obtained by repeatedly applying the $(3,3)$-reduction to the trees shown
in~\cite[Figure 2]{tightkernel}, whose TBR distance is there proven to be exactly $k$. Given that the $(3,3)$-reduction is TBR-preserving, the claim follows.}
}.

\begin{proposition}
For $k\geq 4$, let $T_k$ and $T_k'$, be the two exhaustively reduced unrooted binary phylogenetic trees on $X$ that are shown in Figure~\ref{fig:tight}. Then $d_\TBR(T_k,T_k')=k$.
\end{proposition}


%

\section{Discussion and future work}

\green{To further lower the $11k-9$ bound using the approach described in this article requires reduction rules to prohibit generator sides from having 4 leaves (and 1 or 2 breakpoints), or 3 leaves (and 0 breakpoints). However in such situations it is neither clear how to reduce the TBR distance by 1, or reduce the number of taxa without altering the TBR distance. Hence, new techniques are required which do not just look ``locally'' at individual sides of the generator, but at the way multiple sides of the generator interact. We hope to return to this issue in future work.  Interestingly, although it is not yet clear how to eliminate these cases in the context of kernelization, the analysis in our paper does convey additional structural information.
For example, the argument behind \fuchsia{the $(3,3)$- and $(3,2)$-reduction} directly identifies an edge, in one of the trees, that can safely be deleted if we wish to progressively transform that tree into a maximum agreement forest. These edges can sometimes still be identified even in situations when our new reduction rules do not apply.
Such insights, together with Theorem \ref{thm:allchainsintact}, can potentially be used by FPT branching algorithms that compute the TBR distance by iteratively deleting edges to obtain agreement forests (see e.g. \cite{chen2015parameterized}).
Could the unrooted generator approach, coupled with the reduction rules described in this article, be used to reduce the branching factor of such algorithms?}

\appendix

\begin{figure}[t]
\center
\scalebox{1.2}{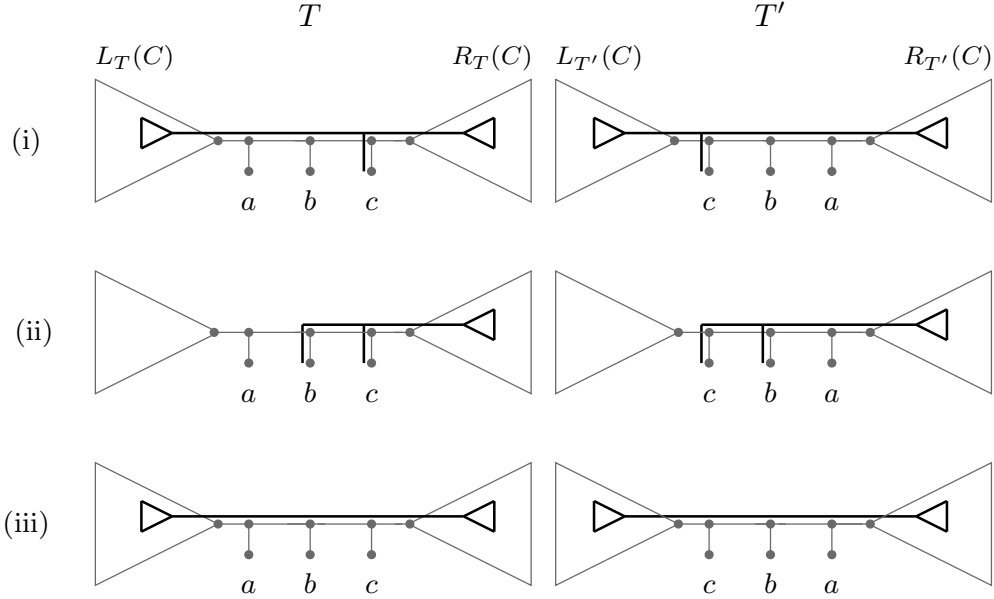}
\caption{\fuchsia{Two unrooted binary phylogenetic trees $T$ and $T'$ on $X$ that have a common 3-chain $C=(a,b,c)$ that is not pendant in $T$ or $T'$. The leaf sets of the subtrees indicated by the left and right solid \redred{grey} triangle of $T$ (resp. $T'$) are $L_T(C)$ and $R_T(C)$ (resp. $L_{T'}(C)$ and $R_{T'}(C)$). (i) An inside-outside component $B$ with respect to $C$ that straddles $C$ in $T$ and $T'$, and with $c\in B$. (ii) An inside-outside component $B'$ with respect to $C$ that does not straddle $C$ in $T$ or $T'$, and with $b,c\in B'$. (iii) A bypass component $B''$ in $T$ and $T'$ with respect to $C$. The components $B$, $B'$, and $B''$ are indicated by their embeddings in $T$ and $T'$ (\blueblue{thick black lines}). 
%
 }}
\label{fig:chain-prooft}
\end{figure}

\section{Proof of Theorem \ref{thm:allchainsintact}}
\begin{proof}
Let $F$ be an arbitrary maximum agreement forest of $T$ and $T'$. Let $C \in K$ be a chain as described in the statement of the theorem (i.e. it has length at length 3, or it has length 2 and it is pendant in at least one of $T$ and $T'$.) For shorthand we call these \emph{eligible} chains. Suppose that $C$ is split in $F$. We will show how to transform $F$ into a new agreement forest $F'$, without increasing the number of components, such that $C$ is preserved in $F'$ and such that all eligible chains that were preserved in $F$ are also preserved in $F'$. Iterating this process will eventually bring us to a maximum agreement forest $F'$ with the desired property, and the proof will be complete. 
It is helpful to recall that
all the chains in $K$ are mutually taxa disjoint, and thus (by Observation \ref{obs:disjoint}) their embeddings are mutually vertex disjoint in both $T$ and $T'$.  \\

Let $J = \{B \in F : C \cap B \neq \emptyset \}$.
We have assumed that $C$ is split, so $|J| \geq 2$. If $B \in J$ and, additionally, $B \cap (X \setminus C) \neq \emptyset$, we call $B$ an \emph{inside-outside component with respect to $C$}. There can be at most 2 such components because a chain connects to the surrounding tree in (at most) 2 places.
If $C$ is not pendant in $T$, then deleting $C$ from $T$ naturally partitions $X \setminus C$ into two disjoint non-empty
sets $L_T(C)$ and $R_T(C)$. Informally these are the taxa in $T$ that are to the ``left'' and ``right'' of $C$. For the purpose of this proof it does not matter which side we designate as left and right. \newblue{If $C$ is not pendant in $T$, then} we say that a component $B \in F$ \emph{straddles $C$ in $T$} if $L_T(C) \cap B \neq \emptyset$ and $R_T(C) \cap B \neq \emptyset$. \newblue{The straddling relation only applies to non-pendant chains: if $C$ is pendant in $T$ then, by definition, it is not possible
for a component of any agreement forest to straddle it in $T$.}\\

We say that a component  $B \in F$ is a \emph{bypass component in $T$ with respect to $C$}, if $B \not \in J$ (i.e. $B \cap C = \emptyset$) and $B$ straddles $C$ in $T$. Informally, $B$ passes ``through'' $C$ in $T$ without including any of its taxa. If in a given context it does not matter whether $B$ is a bypass component in $T$ and/or $T'$, we simply say that $B$ is a bypass component with respect to $C$. \newblue{Figures \ref{fig:chain-prooft} and \ref{fig:chain-prooft2} illustrate a number of these concepts.}  Note that, in the main proof below, we will need to consider the possibility that $B$ (\fuchsia{resp.} $B'$) is a bypass component in $T$ (\fuchsia{resp.} $T'$) with respect to $C$, but that $B \neq B'$, or that a component $B$ is a bypass component in $T$ \emph{and} $T'$ with respect to $C$.   
 The \fuchsia{following observations} will
also
 be useful, for which we omit proofs.

\begin{observation}
\label{obs:useful}

Let $C$, $F$ and $J$ be as defined above.
\begin{enumerate}
\item[(a)] $F$ can contain at most two bypass components with respect to $C$ i.e. one per tree. If $B \in F$ is a bypass component in $T$ \underline{and} $T'$ with respect to $C$, then $B$ is the only bypass component in \fuchsia{$F$} with respect to $C$.
\item[(b)] If $F$ contains at least one bypass component with respect to $C$, then $C$ is atomized in $F$.
\item[(c)] If $F$ contains at least one inside-outside component with respect to $C$, it cannot contain any bypass components with respect to $C$.
\item[(d)] A component $B \in J$ that is not an inside-out component with respect to $C$, has the property $B \subseteq C$.
\item[(e)] If $C$ is pendant in $T$ and/or $T'$, then $F$ contains at most one bypass component with respect to $C$, and at most one inside-outside component with respect to $C$.
\end{enumerate}
\end{observation}
\noindent

 \begin{figure}[t]
\center
\scalebox{1.2}{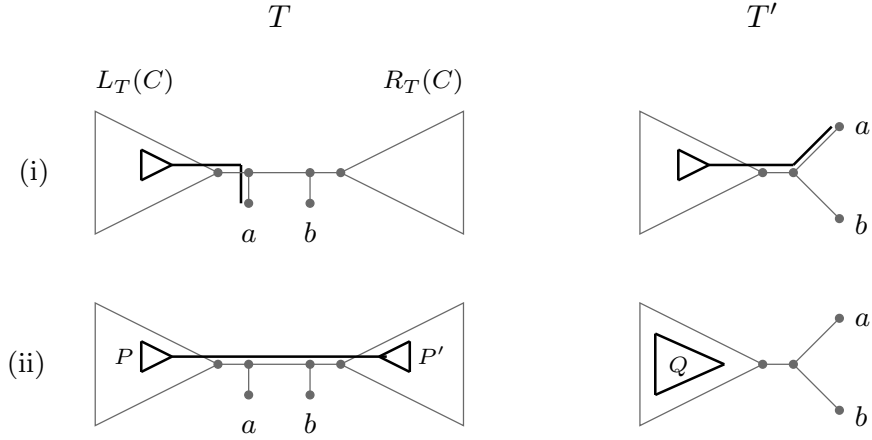}
\caption{\fuchsia{Two unrooted binary phylogenetic trees $T$ and $T'$ on $X$ that have a common 2-chain $C=(a,b)$ that is pendant in $T'$. The leaf set of the subtree indicated by the left and right \redred{grey} solid triangle of $T$ is $L_T(C)$ and $R_T(C)$, respectively, whereas the leaf set of the subtree indicated by the solid \redred{grey} triangle of $T'$ is $X\setminus\{a,b\}$.
(i) An inside-outside component $B$ with respect to $C$ that does not straddle $C$ in $T$, and with \blueblue{$a\in B$}. (ii) A bypass component $B'$ in $T$ with respect to $C$, where the leaf sets of the subtrees indicated by the \blueblue{thick black} triangles are so that $Q=P\cup P'$. The components $B$ and $B'$ are indicated by their embeddings in $T$ and $T'$ (\blueblue{thick black lines}).}}
\label{fig:chain-prooft2}
\end{figure}

We now start with the main proof. We distinguish several (sub)cases.

\begin{enumerate}
\item[1.] \textbf{$C$ is pendant in neither $T$ nor $T'$.} In this case, $|C| \geq 3$, because chains of length 2 are assumed to be pendant in at least one tree. 

\begin{enumerate}
\item[1.1.] Suppose that $F$ contains at least one bypass component with respect to $C$.
Then $C$ is atomized, by Observation \ref{obs:useful}(b). Now, recall Observation \ref{obs:useful}(a). We start by splitting the bypass component(s), as follows. If $F$ contains a bypass component $B$ such that $B$ is a bypass component
in $T$ (but not in $T'$) with respect to $C$, we replace $B$ by two  components $L_T(C) \cap B$ and $R_T(C) \cap B$.
Next, if $F$ contains a bypass component $B'$ such that $B'$ is a bypass component
in $T'$ (but not in $T$) with respect to $C$, we replace $B'$ by two  components $L_{T'}(C) \cap B'$ and $R_{T'}(C) \cap B'$.
A third possibility (which can only hold if neither of the two previous possibilities holds -- see the second part of Observation \ref{obs:useful}(a)) is that $F$ contains a bypass component $B$ that is a bypass component in both $T$ \emph{and} $T'$ with respect to $C$. In this case, we
replace $B$ with non-empty components from the following list.
\begin{itemize}
\item $(L_T(C) \cap B) \cap (L_T'(C) \cap B)$,
\item $(L_T(C) \cap B) \cap (R_T'(C) \cap B)$,
\item $(R_T(C) \cap B) \cap (L_T'(C) \cap B)$,
\item $(R_T(C) \cap B) \cap (R_T'(C) \cap B)$.
\end{itemize}
This captures the situation when we split the same component twice, because it bypasses $C$ in both trees, rather than splitting two distinct components each once. Crucially, at most 3 of these sets can be non-empty. (If all four were non-empty, then this would contradict the $T|B = T'|B$ property of agreement forests\footnote{Essentially we are deleting two edges in $T|B = T'|B$. These two edges induce what in standard phylogenetic terminology are called ``compatible splits''; they are compatible because the two edges are drawn from the same tree. Two splits are compatible if and only if at most 3 of the 4 described intersections are non-empty \cite{SempleSteel2003}.}).
Having split the bypass component(s), we next remove all components $\{x\}$ where $x \in C$, and introduce $C$ as a single component.
Splitting the bypass component(s) increases the number of components by at most 2, but replacing the singleton components with $C$ reduces the number of components by at least 2, because $|C| \geq 3$, so we still have an optimal agreement forest. 

Now, suppose for the sake of \fuchsia{a} contradiction that a previously preserved eligible chain $D$ is split by the modifications described above. Then, at least one of the following holds: (i) $D \cap L_T(C) \neq \emptyset$ and $D \cap R_T(C) \neq \emptyset$,
(ii) $D \cap L_{T'}(C) \neq \emptyset$ and $D \cap R_{T'}(C) \neq \emptyset$. However, if (i) holds then
$T[D] \cap T[C] \neq \emptyset$, and if (ii) holds then $T'[D] \cap T'[C] \neq \emptyset$, both of which contradict
Observation \ref{obs:disjoint}.

\begin{figure}[t]
\center
\scalebox{1}{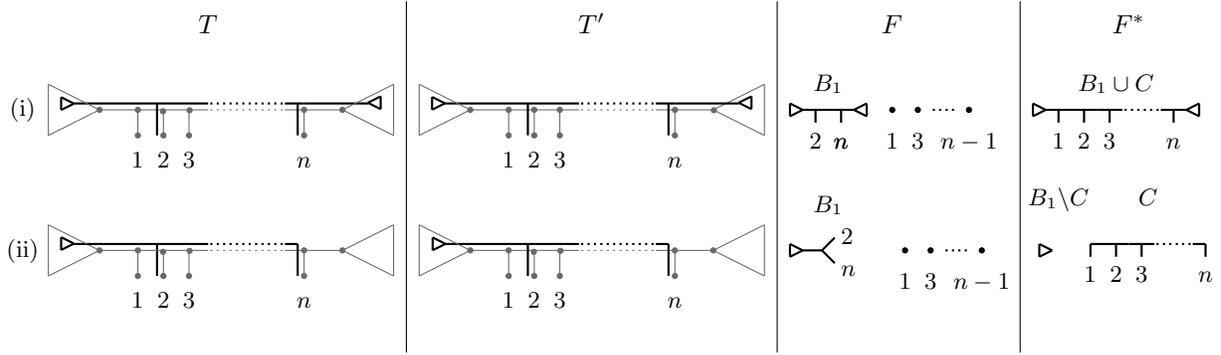}
\caption{\blueblue{Setting as described in Case 1.2.1 in the proof of Theorem~\ref{thm:allchainsintact}, where we consider a chain $C=(1,2,\ldots,n)$ with $n\geq 3$ that is not pendant in $T$ or $T'$, and an agreement forest $F$ for $T$ and $T'$ that does not contain a bypass component and does contain exactly one inside-outside component $B_1$ with respect to $C$ such that $|B_1\cup C|\geq 2$: (i) $B_1$ straddles $C$ in $T$ or $T'$ (and hence $B_1$ straddles $C$ in \redred{both} $T$ and $T'$); (ii) $B_1$ does not straddle $C$ in $T$ or $T'$. The last column shows an agreement forest $F^*$ for $T$ and $T'$ with $|F^*|\leq |F|$ such that $C$ is preserved in $F^*$. Note that only those elements of $F$ and $F^*$ are shown that contain leaves labeled by elements in $C$ while all other elements are omitted since they are the same in $F^*$ and $F$. Furthermore, thick black lines in $T$ and $T'$ indicate the embedding of $B_1$.}}
\label{fig:1pt2pt1}
\end{figure}

 \item[1.2.] Suppose that $F$ does not contain any bypass components with respect to $C$. We now look at the number of inside-outside components with respect to $C$. If there are 0 inside-outside components, then by  Observation \ref{obs:useful}(d) all $B \in J$ have the property $B \subseteq C$. We remove all the $\geq 2$ components in $J$, and replace them with $C$, yielding a valid agreement forest with strictly fewer components than $F$, and thus a contradiction. If there are exactly 2 inside-outside components $B_1, B_2$, then we do the following to $F$: remove $B_1$ \fuchsia{and} $B_2$, discard any components
$B_i$ such that $B_i \subseteq C$ and then finally add the single component $B_1 \cup B_2 \cup C$. This yields a smaller agreement forest and thus also a contradiction.


The only subcase that remains is that there is exactly one inside-outside component $B_1 \in J$.  \redred{We will illustrate
this subcase with a number of additional figures}.

\begin{enumerate}
\item[1.2.1.] Suppose \fuchsia{that} $|B_1 \cap C| \geq 2$.
Observe that if $B_1$ straddles $C$ in at least one of $T$ and $T'$, then (i)
the taxa in $C \setminus B_1$ are all singleton components in $F$ and (ii) $B_1$ actually straddles $C$ in \emph{both} $T$ and $T'$ (because $B_1 \cap C$ is not pendant in $B_1$). \redred{This situation is illustrated in Figure \ref{fig:1pt2pt1}(i).} We remove $B_1$, discard all the singleton components $\{x\}$ such that $x \in C \setminus B_1$, then add the component $B_1 \cup C$. This is a valid agreement forest because the (at least) two taxa in $B_1 \cap C$, combined with the fact that $B_1$ straddles $C$ in both trees, ensure that $T |  (B_1 \cup C) = T' | (B_1 \cup C)$. Noting that $| C \setminus B_1 | \geq 1$ (because $|J| \geq 2$), we thus obtain a smaller agreement forest and thus a contradiction on the optimality \fuchsia{of $F$}. Continuing, suppose that $B_1$ straddles $C$ in neither $T$ nor $T'$; \redred{this situation is illustrated in Figure \ref{fig:1pt2pt1}(ii)}. Informally this means that in both $T$ and $T'$ the inside-outside component $B_1$ enters the chain from only one side. We replace $B_1$ with $B_1 \setminus C$,
delete all components $B_i \subseteq C$ (there is at least one such component, because $|J| \geq 2$ and $B_1$ is the only inside-outside component), and introduce component $C$. The overall size of the agreement forest does not increase. This cannot split any previously preserved eligible chain $D$, because any such chain $D$ would have taxa in both $C$ and $B_1 \setminus C$, yielding $T[C] \cap T[D] \neq \emptyset$ and a contradiction to Observation \ref{obs:disjoint}.

\begin{figure}[t]
\center
\scalebox{1}{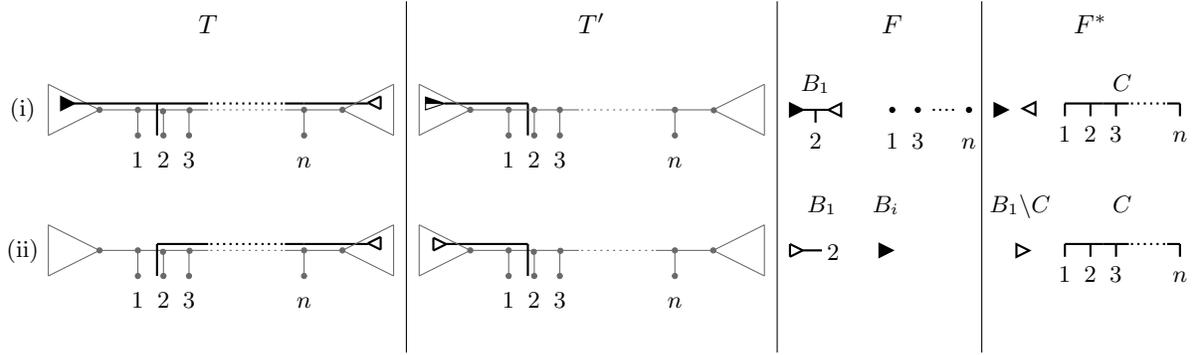}
\caption{\blueblue{Setting as described in Case 1.2.2 in the proof of Theorem~\ref{thm:allchainsintact}, where we consider a chain $C=(1,2,\ldots,n)$ with $n\geq 3$ that is not pendant in $T$ or $T'$, and an agreement forest $F$ for $T$ and $T'$ that does not contain a bypass component and does contain exactly one inside-outside component $B_1$ with respect to $C$ such that $|B_1\cap C|= 1$:  (i) $B_1$ straddles $C$ in one of $T$ and $T'$, say $T$; (ii) $B_1$ does not straddle $C$ in $T$ or $T'$ which implies that there exists at least one component $B_i$ with $B_i\subseteq C$. The last column shows an agreement forest $F^*$ for $T$ and $T'$ with $|F^*|\leq |F|$ such that $C$ is preserved in $F^*$. Note that only those elements of $F$ and $F^*$ are shown that contain leaves labeled by elements in $C$ while all other elements are omitted since they are the same in $F^*$ and $F$. Furthermore, thick black lines in $T$ and $T'$ indicate the embedding of $B_1$.}}
\label{fig:1pt2pt2}
\end{figure}
 
\item[1.2.2.] Suppose \fuchsia{that} $|B_1 \cap C| = 1$.
Let $x$ be the unique taxon in $B_1 \cap C$. Suppose $B_1$ straddles $C$ in at least one of $T$ and $T'$; assume without loss of generality that it is in $T$ (\redred{see Figure \ref{fig:1pt2pt2}(i)}). Then the $\geq 2$ taxa in $C \setminus \{x\}$ must be singleton components in $F$. We delete $B_1$, delete the at least 2 singleton components formed by taxa in $C \setminus \{x\}$, and introduce components $L_T(C), R_T(C), C$. This does not increase the number
of components in the agreement forest, so it is still a maximum agreement forest. As usual, the only way that a previously preserved eligible chain could be split is if it contains taxa from both $L_T(C)$ and $R_T(C)$ which 
contradicts Observation \ref{obs:disjoint}. Finally, suppose \fuchsia{that} $B_1$ straddles $C$ in neither $T$ nor $T'$ (\redred{see Figure \ref{fig:1pt2pt2}(ii)}). We
\newblue{replace $B_1$ with $B_1 \setminus C$}, delete all components $B_i \subseteq C$ (there will be at least one such component \blueblue{because $B_1$ is the only inside-outside component with respect to $C$}), and introduce
\newblue{$C$. The overall size of the agreement forest does not increase. This cannot split any previously preserved eligible chain $D$, because then $T[D] \cap T[C] \neq \emptyset$, contradicting Observation \ref{obs:disjoint}.}
\end{enumerate}
 
\end{enumerate}

\begin{figure}[t]
\center
\scalebox{1}{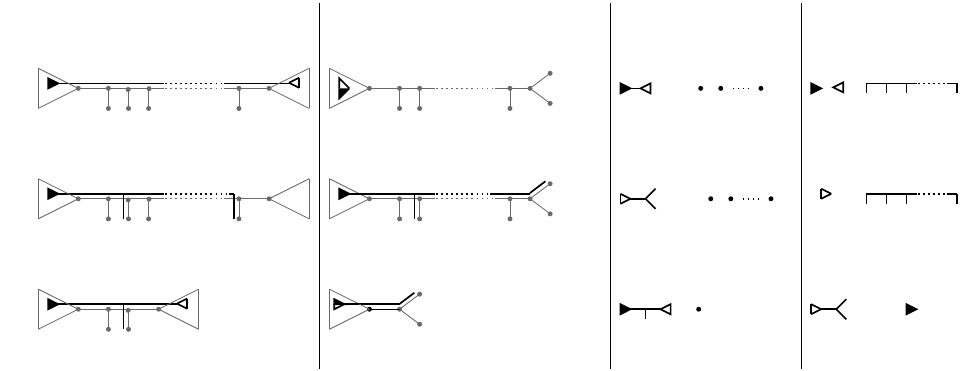}
\caption{\blueblue{Setting as described in Case 2 in the proof of Theorem~\ref{thm:allchainsintact}, where we consider a chain $C=(1,2,\ldots,n)$ with $n\geq 2$ that is pendant in one of $T$ or $T'$, say $T'$, and an agreement forest $F$ for $T$ and $T'$: (i) $F$ contains a bypass component $B$ in $T$ with respect to $C$; (ii) $F$ contains no bypass component with respect to $C$ but a unique inside-outside component $B$ with respect to $C$ such that $|B\cap C|\geq 2$ (and hence $n\geq 3$) and $B$ does not straddle $C$ in $T$ or $T'$; (iii) $F$ contains no bypass component with respect to $C$ but a unique inside-outside component $B$ with respect to $C$ such that $|B\cap C|=1$ (and hence $C=(1,2)$)  and $B$ straddles $C$ in $T$.
The last column shows an agreement forest $F^*$ for $T$ and $T'$ with $|F^*|\leq |F|$ such that $C$ is preserved in $F^*$. Note that only those elements of $F$ and $F^*$ are shown that contain leaves labeled by elements in $C$ while all other elements are omitted since they are the same in $F^*$ and $F$. Furthermore, thick black lines in $T$ and $T'$ indicate the embedding of $B$.}}
\label{fig:case2}
\end{figure}

\item[2.] \textbf{$C$ is pendant in at least one of $T$ and $T'$.} \fuchsia{In this case, we have $|C|\geq 2$}. Assume without loss of generality that $C$ is pendant in $T'$. Recall from Observation \ref{obs:useful}(e) that $F$ contains at most one inside-outside component with respect to $C$, at most one bypass component with respect to $C$, and that from Observations \ref{obs:useful}(b) and (c) at most one of these two situations can hold. Suppose that $B \in F$ is a bypass component with respect to $C$;
this is necessarily in $T$, since $B$ cannot be a bypass component in $T'$ with respect to $C$ (due to pendancy). \redred{Figure \ref{fig:case2}(i) illustrates this situation.} By Observation \ref{obs:useful}(b), $C$ is atomized. Now, consider the construction in \fuchsia{Case} 1.1. Here we only have one
bypass component to split, but on the other hand we are only allowed to use the weaker assumption $|C| \geq 2$. These
two cancel each other out, so Case 1.1 still goes through. So henceforth we can assume that there are no bypass components with respect to $C$.
If there are no inside-outside components with respect to $C$ then replacing the components in $J$ (which are all subsets of $C$) with $C$ reduces the overall number of components, because $|J| \geq 2$, immediately yielding a contradiction to the optimality \fuchsia{of $F$}. So let $B \in F$ be the unique inside-outside component in $J$. If $|B \cap C| \geq 2$ (note that if $|C|=2$ then this cannot happen, because it would imply $|J|=1$) then, due to the pendancy of $C$ in $T'$, $B$ straddles $C$ in neither $T$ nor $T'$; \redred{this is the situation shown in Figure \ref{fig:case2}(ii).}  (\newblue{The fact that $B$ does not straddle $C$ in $T'$ is automatic, since pendant chains cannot be straddled, by definition. The same holds if $C$ is pendant in $T$. If $C$ is not pendant in $T$, observe that if $B$ straddled $C$ in $T$, then $B \cap C$ would not be pendant in $T|B$. However, the fact that
$C$ is pendant in $T'$ means that $B \cap C$ must be pendant in $T'|B$. Taken together we would have $T|B \neq T'|B$, contradicting the assumption
that $B$ is a component of an agreement forest of $T$ and $T'$.)}

We replace $B$ with $B \setminus C$, delete all components $B_i \subseteq C$ (there must be at least one such component), and introduce $C$ as a component. Thus, the overall size of the agreement forest does not increase. This cannot split any previously preserved eligible chain $D$, because (by the usual argument) any such chain $D$ would have taxa in both $C$ and $B\setminus C$, and this contradicts Observation \ref{obs:disjoint}. So assume that $|B \cap C| = 1$. If $|C| \geq 3$, Case 1.2.2 goes through unchanged. If $|C|=2$,
then Case 1.2.2 mostly still holds, except for one situation: when $B_1$ straddles $C$ in (say) $T$. \redred{This is illustrated in Figure \ref{fig:case2}(iii).} The problem here is that $C \setminus \{x\}$ (where $x$ is as 
defined as in that case) contains only 1 taxon, so we introduce more components than we delete. However, we can modify the argument as follows. Let $y$ be the unique taxon in $C$ that is not equal to $x$. We delete \fuchsia{$B$ and} $\{y\}$, and introduce components $L_T(C) \cup \{x,y\}, R_T(C)$. Hence, the agreement forest does not increase in size. A previously preserved eligible chain $D$ cannot be split by this modification, since it would imply that $D$ contains taxa from both $L_T(C)$ and $R_T(C)$, yielding the usual contradiction to Observation \ref{obs:disjoint}.
\end{enumerate}
\qed
\end{proof}

\noindent{\bf Acknowledgements.}
\fuchsia{The second author was supported by the New Zealand Marsden Fund. Both authors would also like to thank the reviewers for their insightful comments.}


\bibliography{NewRulesForTBR-Accepted_Version}{}
\bibliographystyle{plain}

\end{document}

%% file: trees.pdf_t
\begin{picture}(0,0)%
\includegraphics{trees.pdf}%
\end{picture}%
\setlength{\unitlength}{3522sp}%
\begingroup\makeatletter\ifx\SetFigFont\undefined%
\gdef\SetFigFont#1#2#3#4#5{%
  \reset@font\fontsize{#1}{#2pt}%
  \fontfamily{#3}\fontseries{#4}\fontshape{#5}%
  \selectfont}%
\fi\endgroup%
\begin{picture}(3585,931)(2956,-4850)
\put(6526,-4651){\makebox(0,0)[lb]{\smash{{\SetFigFont{9}{10.8}{\rmdefault}{\mddefault}{\updefault}$d$}}}}
\put(4231,-4651){\makebox(0,0)[lb]{\smash{{\SetFigFont{9}{10.8}{\rmdefault}{\mddefault}{\updefault}$d$}}}}
\put(5266,-4066){\makebox(0,0)[rb]{\smash{{\SetFigFont{9}{10.8}{\rmdefault}{\mddefault}{\updefault}$e$}}}}
\put(5266,-4651){\makebox(0,0)[rb]{\smash{{\SetFigFont{9}{10.8}{\rmdefault}{\mddefault}{\updefault}$b$}}}}
\put(5896,-4786){\makebox(0,0)[b]{\smash{{\SetFigFont{9}{10.8}{\rmdefault}{\mddefault}{\updefault}$c$}}}}
\put(6526,-4066){\makebox(0,0)[lb]{\smash{{\SetFigFont{9}{10.8}{\rmdefault}{\mddefault}{\updefault}$a$}}}}
\put(2971,-4066){\makebox(0,0)[rb]{\smash{{\SetFigFont{9}{10.8}{\rmdefault}{\mddefault}{\updefault}$a$}}}}
\put(2971,-4651){\makebox(0,0)[rb]{\smash{{\SetFigFont{9}{10.8}{\rmdefault}{\mddefault}{\updefault}$b$}}}}
\put(3601,-4786){\makebox(0,0)[b]{\smash{{\SetFigFont{9}{10.8}{\rmdefault}{\mddefault}{\updefault}$c$}}}}
\put(4231,-4066){\makebox(0,0)[lb]{\smash{{\SetFigFont{9}{10.8}{\rmdefault}{\mddefault}{\updefault}$e$}}}}
\end{picture}%

%% file: tbr.pdf_t
\begin{picture}(0,0)%
\includegraphics{tbr.pdf}%
\end{picture}%
\setlength{\unitlength}{3522sp}%
\begingroup\makeatletter\ifx\SetFigFont\undefined%
\gdef\SetFigFont#1#2#3#4#5{%
  \reset@font\fontsize{#1}{#2pt}%
  \fontfamily{#3}\fontseries{#4}\fontshape{#5}%
  \selectfont}%
\fi\endgroup%
\begin{picture}(7635,1419)(2956,-5269)
\put(9946,-4156){\makebox(0,0)[b]{\smash{{\SetFigFont{9}{10.8}{\rmdefault}{\mddefault}{\updefault}$v_2$}}}}
\put(5581,-4066){\makebox(0,0)[rb]{\smash{{\SetFigFont{9}{10.8}{\rmdefault}{\mddefault}{\updefault}$a$}}}}
\put(5581,-4651){\makebox(0,0)[rb]{\smash{{\SetFigFont{9}{10.8}{\rmdefault}{\mddefault}{\updefault}$b$}}}}
\put(6616,-4651){\makebox(0,0)[lb]{\smash{{\SetFigFont{9}{10.8}{\rmdefault}{\mddefault}{\updefault}$c$}}}}
\put(6616,-4066){\makebox(0,0)[lb]{\smash{{\SetFigFont{9}{10.8}{\rmdefault}{\mddefault}{\updefault}$d$}}}}
\put(2971,-4066){\makebox(0,0)[rb]{\smash{{\SetFigFont{9}{10.8}{\rmdefault}{\mddefault}{\updefault}$a$}}}}
\put(2971,-4651){\makebox(0,0)[rb]{\smash{{\SetFigFont{9}{10.8}{\rmdefault}{\mddefault}{\updefault}$b$}}}}
\put(3601,-4786){\makebox(0,0)[b]{\smash{{\SetFigFont{9}{10.8}{\rmdefault}{\mddefault}{\updefault}$c$}}}}
\put(3871,-4786){\makebox(0,0)[b]{\smash{{\SetFigFont{9}{10.8}{\rmdefault}{\mddefault}{\updefault}$d$}}}}
\put(4141,-4786){\makebox(0,0)[b]{\smash{{\SetFigFont{9}{10.8}{\rmdefault}{\mddefault}{\updefault}$e$}}}}
\put(4771,-4651){\makebox(0,0)[lb]{\smash{{\SetFigFont{9}{10.8}{\rmdefault}{\mddefault}{\updefault}$f$}}}}
\put(4771,-4066){\makebox(0,0)[lb]{\smash{{\SetFigFont{9}{10.8}{\rmdefault}{\mddefault}{\updefault}$g$}}}}
\put(3871,-4201){\makebox(0,0)[b]{\smash{{\SetFigFont{9}{10.8}{\rmdefault}{\mddefault}{\updefault}$u_1$}}}}
\put(4141,-4201){\makebox(0,0)[b]{\smash{{\SetFigFont{9}{10.8}{\rmdefault}{\mddefault}{\updefault}$u_2$}}}}
\put(3871,-5191){\makebox(0,0)[b]{\smash{{\SetFigFont{10}{12.0}{\rmdefault}{\mddefault}{\updefault}$T$}}}}
\put(9721,-5191){\makebox(0,0)[b]{\smash{{\SetFigFont{10}{12.0}{\rmdefault}{\mddefault}{\updefault}$T'$}}}}
\put(6076,-5191){\makebox(0,0)[b]{\smash{{\SetFigFont{10}{12.0}{\rmdefault}{\mddefault}{\updefault}$T_1$}}}}
\put(7606,-5191){\makebox(0,0)[b]{\smash{{\SetFigFont{10}{12.0}{\rmdefault}{\mddefault}{\updefault}$T_2$}}}}
\put(5896,-4156){\makebox(0,0)[lb]{\smash{{\SetFigFont{9}{10.8}{\rmdefault}{\mddefault}{\updefault}$v_1$}}}}
\put(7831,-4426){\makebox(0,0)[lb]{\smash{{\SetFigFont{9}{10.8}{\rmdefault}{\mddefault}{\updefault}$v_2$}}}}
\put(9676,-4156){\makebox(0,0)[b]{\smash{{\SetFigFont{9}{10.8}{\rmdefault}{\mddefault}{\updefault}$v_1$}}}}
\put(7966,-4651){\makebox(0,0)[lb]{\smash{{\SetFigFont{9}{10.8}{\rmdefault}{\mddefault}{\updefault}$f$}}}}
\put(7966,-4066){\makebox(0,0)[lb]{\smash{{\SetFigFont{9}{10.8}{\rmdefault}{\mddefault}{\updefault}$g$}}}}
\put(7246,-4336){\makebox(0,0)[rb]{\smash{{\SetFigFont{9}{10.8}{\rmdefault}{\mddefault}{\updefault}$e$}}}}
\put(8776,-4021){\makebox(0,0)[rb]{\smash{{\SetFigFont{9}{10.8}{\rmdefault}{\mddefault}{\updefault}$c$}}}}
\put(8776,-4606){\makebox(0,0)[rb]{\smash{{\SetFigFont{9}{10.8}{\rmdefault}{\mddefault}{\updefault}$d$}}}}
\put(9406,-4741){\makebox(0,0)[b]{\smash{{\SetFigFont{9}{10.8}{\rmdefault}{\mddefault}{\updefault}$b$}}}}
\put(9676,-4741){\makebox(0,0)[b]{\smash{{\SetFigFont{9}{10.8}{\rmdefault}{\mddefault}{\updefault}$a$}}}}
\put(9946,-4741){\makebox(0,0)[b]{\smash{{\SetFigFont{9}{10.8}{\rmdefault}{\mddefault}{\updefault}$f$}}}}
\put(10576,-4606){\makebox(0,0)[lb]{\smash{{\SetFigFont{9}{10.8}{\rmdefault}{\mddefault}{\updefault}$e$}}}}
\put(10576,-4021){\makebox(0,0)[lb]{\smash{{\SetFigFont{9}{10.8}{\rmdefault}{\mddefault}{\updefault}$g$}}}}
\end{picture}%

%% file: reduction.pdf_t
\begin{picture}(0,0)%
\includegraphics{reduction.pdf}%
\end{picture}%
\setlength{\unitlength}{3522sp}%
\begingroup\makeatletter\ifx\SetFigFont\undefined%
\gdef\SetFigFont#1#2#3#4#5{%
  \reset@font\fontsize{#1}{#2pt}%
  \fontfamily{#3}\fontseries{#4}\fontshape{#5}%
  \selectfont}%
\fi\endgroup%
\begin{picture}(7374,9834)(2449,-12133)
\put(6706,-11806){\makebox(0,0)[b]{\smash{{\SetFigFont{8}{9.6}{\rmdefault}{\mddefault}{\updefault}$P'$}}}}
\put(6706,-11311){\makebox(0,0)[b]{\smash{{\SetFigFont{9}{10.8}{\rmdefault}{\mddefault}{\updefault}$z$}}}}
\put(6706,-10141){\makebox(0,0)[b]{\smash{{\SetFigFont{8}{9.6}{\rmdefault}{\mddefault}{\updefault}$P$}}}}
\put(6976,-10501){\makebox(0,0)[lb]{\smash{{\SetFigFont{9}{10.8}{\rmdefault}{\mddefault}{\updefault}$a$}}}}
\put(6976,-10681){\makebox(0,0)[lb]{\smash{{\SetFigFont{9}{10.8}{\rmdefault}{\mddefault}{\updefault}$b$}}}}
\put(6706,-11041){\makebox(0,0)[b]{\smash{{\SetFigFont{9}{10.8}{\rmdefault}{\mddefault}{\updefault}$c$}}}}
\put(4951,-11581){\makebox(0,0)[b]{\smash{{\SetFigFont{8}{9.6}{\rmdefault}{\mddefault}{\updefault}$Q'$}}}}
\put(4951,-10141){\makebox(0,0)[b]{\smash{{\SetFigFont{8}{9.6}{\rmdefault}{\mddefault}{\updefault}$Q$}}}}
\put(5221,-10501){\makebox(0,0)[lb]{\smash{{\SetFigFont{9}{10.8}{\rmdefault}{\mddefault}{\updefault}$a$}}}}
\put(5221,-10681){\makebox(0,0)[lb]{\smash{{\SetFigFont{9}{10.8}{\rmdefault}{\mddefault}{\updefault}$b$}}}}
\put(5221,-10861){\makebox(0,0)[lb]{\smash{{\SetFigFont{9}{10.8}{\rmdefault}{\mddefault}{\updefault}$c$}}}}
\put(5221,-11041){\makebox(0,0)[lb]{\smash{{\SetFigFont{9}{10.8}{\rmdefault}{\mddefault}{\updefault}$y$}}}}
\put(5221,-11221){\makebox(0,0)[lb]{\smash{{\SetFigFont{9}{10.8}{\rmdefault}{\mddefault}{\updefault}$z$}}}}
\put(2611,-3346){\makebox(0,0)[b]{\smash{{\SetFigFont{9}{10.8}{\rmdefault}{\mddefault}{\updefault}(i)}}}}
\put(6661,-3346){\makebox(0,0)[b]{\smash{{\SetFigFont{8}{9.6}{\rmdefault}{\mddefault}{\updefault}$P$}}}}
\put(8416,-3346){\makebox(0,0)[b]{\smash{{\SetFigFont{8}{9.6}{\rmdefault}{\mddefault}{\updefault}$Q$}}}}
\put(7561,-4786){\makebox(0,0)[b]{\smash{{\SetFigFont{9}{10.8}{\rmdefault}{\mddefault}{\updefault}$b$}}}}
\put(7786,-4426){\makebox(0,0)[lb]{\smash{{\SetFigFont{9}{10.8}{\rmdefault}{\mddefault}{\updefault}$c$}}}}
\put(7336,-4786){\makebox(0,0)[b]{\smash{{\SetFigFont{9}{10.8}{\rmdefault}{\mddefault}{\updefault}$a$}}}}
\put(6661,-4426){\makebox(0,0)[b]{\smash{{\SetFigFont{8}{9.6}{\rmdefault}{\mddefault}{\updefault}$P$}}}}
\put(2611,-4426){\makebox(0,0)[b]{\smash{{\SetFigFont{9}{10.8}{\rmdefault}{\mddefault}{\updefault}(ii)}}}}
\put(3691,-3706){\makebox(0,0)[b]{\smash{{\SetFigFont{9}{10.8}{\rmdefault}{\mddefault}{\updefault}$a$}}}}
\put(3871,-3706){\makebox(0,0)[b]{\smash{{\SetFigFont{9}{10.8}{\rmdefault}{\mddefault}{\updefault}$b$}}}}
\put(4141,-3346){\makebox(0,0)[lb]{\smash{{\SetFigFont{9}{10.8}{\rmdefault}{\mddefault}{\updefault}$c$}}}}
\put(3016,-3346){\makebox(0,0)[b]{\smash{{\SetFigFont{8}{9.6}{\rmdefault}{\mddefault}{\updefault}$P$}}}}
\put(2611,-6046){\makebox(0,0)[b]{\smash{{\SetFigFont{9}{10.8}{\rmdefault}{\mddefault}{\updefault}(iii)}}}}
\put(4951,-5236){\makebox(0,0)[b]{\smash{{\SetFigFont{8}{9.6}{\rmdefault}{\mddefault}{\updefault}$Q$}}}}
\put(5221,-5596){\makebox(0,0)[lb]{\smash{{\SetFigFont{9}{10.8}{\rmdefault}{\mddefault}{\updefault}$a$}}}}
\put(4951,-5956){\makebox(0,0)[b]{\smash{{\SetFigFont{9}{10.8}{\rmdefault}{\mddefault}{\updefault}$b$}}}}
\put(6706,-7036){\makebox(0,0)[b]{\smash{{\SetFigFont{8}{9.6}{\rmdefault}{\mddefault}{\updefault}$P'$}}}}
\put(6976,-6676){\makebox(0,0)[lb]{\smash{{\SetFigFont{9}{10.8}{\rmdefault}{\mddefault}{\updefault}$c$}}}}
\put(6706,-5236){\makebox(0,0)[b]{\smash{{\SetFigFont{8}{9.6}{\rmdefault}{\mddefault}{\updefault}$P$}}}}
\put(6976,-5596){\makebox(0,0)[lb]{\smash{{\SetFigFont{9}{10.8}{\rmdefault}{\mddefault}{\updefault}$a$}}}}
\put(6706,-5956){\makebox(0,0)[b]{\smash{{\SetFigFont{9}{10.8}{\rmdefault}{\mddefault}{\updefault}$b$}}}}
\put(4951,-6181){\makebox(0,0)[b]{\smash{{\SetFigFont{9}{10.8}{\rmdefault}{\mddefault}{\updefault}$x$}}}}
\put(4951,-7036){\makebox(0,0)[b]{\smash{{\SetFigFont{8}{9.6}{\rmdefault}{\mddefault}{\updefault}$Q'$}}}}
\put(8416,-5236){\makebox(0,0)[b]{\smash{{\SetFigFont{8}{9.6}{\rmdefault}{\mddefault}{\updefault}$Q$}}}}
\put(8686,-5596){\makebox(0,0)[lb]{\smash{{\SetFigFont{9}{10.8}{\rmdefault}{\mddefault}{\updefault}$a$}}}}
\put(8416,-5956){\makebox(0,0)[b]{\smash{{\SetFigFont{9}{10.8}{\rmdefault}{\mddefault}{\updefault}$b$}}}}
\put(8416,-7036){\makebox(0,0)[b]{\smash{{\SetFigFont{8}{9.6}{\rmdefault}{\mddefault}{\updefault}$Q'$}}}}
\put(8686,-6676){\makebox(0,0)[lb]{\smash{{\SetFigFont{9}{10.8}{\rmdefault}{\mddefault}{\updefault}$c$}}}}
\put(8416,-6361){\makebox(0,0)[b]{\smash{{\SetFigFont{9}{10.8}{\rmdefault}{\mddefault}{\updefault}$d$}}}}
\put(5221,-6676){\makebox(0,0)[lb]{\smash{{\SetFigFont{9}{10.8}{\rmdefault}{\mddefault}{\updefault}$c$}}}}
\put(5221,-6496){\makebox(0,0)[lb]{\smash{{\SetFigFont{9}{10.8}{\rmdefault}{\mddefault}{\updefault}$d$}}}}
\put(6706,-6361){\makebox(0,0)[b]{\smash{{\SetFigFont{9}{10.8}{\rmdefault}{\mddefault}{\updefault}$d$}}}}
\put(2611,-8611){\makebox(0,0)[b]{\smash{{\SetFigFont{9}{10.8}{\rmdefault}{\mddefault}{\updefault}(iv)}}}}
\put(3061,-10141){\makebox(0,0)[b]{\smash{{\SetFigFont{8}{9.6}{\rmdefault}{\mddefault}{\updefault}$P$}}}}
\put(3331,-10501){\makebox(0,0)[lb]{\smash{{\SetFigFont{9}{10.8}{\rmdefault}{\mddefault}{\updefault}$a$}}}}
\put(3331,-10681){\makebox(0,0)[lb]{\smash{{\SetFigFont{9}{10.8}{\rmdefault}{\mddefault}{\updefault}$b$}}}}
\put(3061,-11041){\makebox(0,0)[b]{\smash{{\SetFigFont{9}{10.8}{\rmdefault}{\mddefault}{\updefault}$c$}}}}
\put(2611,-10951){\makebox(0,0)[b]{\smash{{\SetFigFont{9}{10.8}{\rmdefault}{\mddefault}{\updefault}(v)}}}}
\put(3061,-7576){\makebox(0,0)[b]{\smash{{\SetFigFont{8}{9.6}{\rmdefault}{\mddefault}{\updefault}$P$}}}}
\put(3331,-7936){\makebox(0,0)[lb]{\smash{{\SetFigFont{9}{10.8}{\rmdefault}{\mddefault}{\updefault}$a$}}}}
\put(3331,-8116){\makebox(0,0)[lb]{\smash{{\SetFigFont{9}{10.8}{\rmdefault}{\mddefault}{\updefault}$b$}}}}
\put(3061,-8476){\makebox(0,0)[b]{\smash{{\SetFigFont{9}{10.8}{\rmdefault}{\mddefault}{\updefault}$c$}}}}
\put(3061,-9601){\makebox(0,0)[b]{\smash{{\SetFigFont{8}{9.6}{\rmdefault}{\mddefault}{\updefault}$P'$}}}}
\put(3331,-9241){\makebox(0,0)[lb]{\smash{{\SetFigFont{9}{10.8}{\rmdefault}{\mddefault}{\updefault}$z$}}}}
\put(3331,-9061){\makebox(0,0)[lb]{\smash{{\SetFigFont{9}{10.8}{\rmdefault}{\mddefault}{\updefault}$y$}}}}
\put(3061,-8746){\makebox(0,0)[b]{\smash{{\SetFigFont{9}{10.8}{\rmdefault}{\mddefault}{\updefault}$x$}}}}
\put(3061,-11986){\makebox(0,0)[b]{\smash{{\SetFigFont{8}{9.6}{\rmdefault}{\mddefault}{\updefault}$P'$}}}}
\put(3331,-11626){\makebox(0,0)[lb]{\smash{{\SetFigFont{9}{10.8}{\rmdefault}{\mddefault}{\updefault}$z$}}}}
\put(3061,-11266){\makebox(0,0)[b]{\smash{{\SetFigFont{9}{10.8}{\rmdefault}{\mddefault}{\updefault}$y$}}}}
\put(9271,-4786){\makebox(0,0)[b]{\smash{{\SetFigFont{9}{10.8}{\rmdefault}{\mddefault}{\updefault}$b$}}}}
\put(9496,-4426){\makebox(0,0)[lb]{\smash{{\SetFigFont{9}{10.8}{\rmdefault}{\mddefault}{\updefault}$c$}}}}
\put(9046,-4786){\makebox(0,0)[b]{\smash{{\SetFigFont{9}{10.8}{\rmdefault}{\mddefault}{\updefault}$a$}}}}
\put(8461,-4426){\makebox(0,0)[b]{\smash{{\SetFigFont{8}{9.6}{\rmdefault}{\mddefault}{\updefault}$Q-\{x\}$}}}}
\put(5806,-4786){\makebox(0,0)[b]{\smash{{\SetFigFont{9}{10.8}{\rmdefault}{\mddefault}{\updefault}$b$}}}}
\put(6031,-4426){\makebox(0,0)[lb]{\smash{{\SetFigFont{9}{10.8}{\rmdefault}{\mddefault}{\updefault}$c$}}}}
\put(5581,-4786){\makebox(0,0)[b]{\smash{{\SetFigFont{9}{10.8}{\rmdefault}{\mddefault}{\updefault}$a$}}}}
\put(4906,-4426){\makebox(0,0)[b]{\smash{{\SetFigFont{8}{9.6}{\rmdefault}{\mddefault}{\updefault}$Q$}}}}
\put(5806,-3706){\makebox(0,0)[b]{\smash{{\SetFigFont{9}{10.8}{\rmdefault}{\mddefault}{\updefault}$b$}}}}
\put(6031,-3346){\makebox(0,0)[lb]{\smash{{\SetFigFont{9}{10.8}{\rmdefault}{\mddefault}{\updefault}$a$}}}}
\put(5581,-3706){\makebox(0,0)[b]{\smash{{\SetFigFont{9}{10.8}{\rmdefault}{\mddefault}{\updefault}$c$}}}}
\put(4951,-3346){\makebox(0,0)[b]{\smash{{\SetFigFont{8}{9.6}{\rmdefault}{\mddefault}{\updefault}$Q$}}}}
\put(3691,-4786){\makebox(0,0)[b]{\smash{{\SetFigFont{9}{10.8}{\rmdefault}{\mddefault}{\updefault}$a$}}}}
\put(3871,-4786){\makebox(0,0)[b]{\smash{{\SetFigFont{9}{10.8}{\rmdefault}{\mddefault}{\updefault}$b$}}}}
\put(4051,-4786){\makebox(0,0)[b]{\smash{{\SetFigFont{9}{10.8}{\rmdefault}{\mddefault}{\updefault}$c$}}}}
\put(4321,-4426){\makebox(0,0)[lb]{\smash{{\SetFigFont{9}{10.8}{\rmdefault}{\mddefault}{\updefault}$x$}}}}
\put(3016,-4426){\makebox(0,0)[b]{\smash{{\SetFigFont{8}{9.6}{\rmdefault}{\mddefault}{\updefault}$P$}}}}
\put(6751,-2536){\makebox(0,0)[b]{\smash{{\SetFigFont{10}{12.0}{\rmdefault}{\mddefault}{\updefault}$T_r$}}}}
\put(3016,-2536){\makebox(0,0)[b]{\smash{{\SetFigFont{10}{12.0}{\rmdefault}{\mddefault}{\updefault}$T$}}}}
\put(4906,-2536){\makebox(0,0)[b]{\smash{{\SetFigFont{10}{12.0}{\rmdefault}{\mddefault}{\updefault}$T'$}}}}
\put(8416,-2536){\makebox(0,0)[b]{\smash{{\SetFigFont{10}{12.0}{\rmdefault}{\mddefault}{\updefault}$T_r'$}}}}
\put(3061,-5236){\makebox(0,0)[b]{\smash{{\SetFigFont{8}{9.6}{\rmdefault}{\mddefault}{\updefault}$P$}}}}
\put(3331,-5596){\makebox(0,0)[lb]{\smash{{\SetFigFont{9}{10.8}{\rmdefault}{\mddefault}{\updefault}$a$}}}}
\put(3331,-5776){\makebox(0,0)[lb]{\smash{{\SetFigFont{9}{10.8}{\rmdefault}{\mddefault}{\updefault}$b$}}}}
\put(3061,-6136){\makebox(0,0)[b]{\smash{{\SetFigFont{9}{10.8}{\rmdefault}{\mddefault}{\updefault}$x$}}}}
\put(3061,-7036){\makebox(0,0)[b]{\smash{{\SetFigFont{8}{9.6}{\rmdefault}{\mddefault}{\updefault}$P'$}}}}
\put(3061,-6361){\makebox(0,0)[b]{\smash{{\SetFigFont{9}{10.8}{\rmdefault}{\mddefault}{\updefault}$d$}}}}
\put(3331,-6676){\makebox(0,0)[lb]{\smash{{\SetFigFont{9}{10.8}{\rmdefault}{\mddefault}{\updefault}$c$}}}}
\put(4951,-7576){\makebox(0,0)[b]{\smash{{\SetFigFont{8}{9.6}{\rmdefault}{\mddefault}{\updefault}$Q$}}}}
\put(5221,-7936){\makebox(0,0)[lb]{\smash{{\SetFigFont{9}{10.8}{\rmdefault}{\mddefault}{\updefault}$a$}}}}
\put(5221,-8116){\makebox(0,0)[lb]{\smash{{\SetFigFont{9}{10.8}{\rmdefault}{\mddefault}{\updefault}$b$}}}}
\put(5221,-8296){\makebox(0,0)[lb]{\smash{{\SetFigFont{9}{10.8}{\rmdefault}{\mddefault}{\updefault}$c$}}}}
\put(5221,-8476){\makebox(0,0)[lb]{\smash{{\SetFigFont{9}{10.8}{\rmdefault}{\mddefault}{\updefault}$x$}}}}
\put(4951,-9196){\makebox(0,0)[b]{\smash{{\SetFigFont{8}{9.6}{\rmdefault}{\mddefault}{\updefault}$Q'$}}}}
\put(5221,-8836){\makebox(0,0)[lb]{\smash{{\SetFigFont{9}{10.8}{\rmdefault}{\mddefault}{\updefault}$z$}}}}
\put(5221,-8656){\makebox(0,0)[lb]{\smash{{\SetFigFont{9}{10.8}{\rmdefault}{\mddefault}{\updefault}$y$}}}}
\put(6706,-9241){\makebox(0,0)[b]{\smash{{\SetFigFont{8}{9.6}{\rmdefault}{\mddefault}{\updefault}$P'$}}}}
\put(6706,-8746){\makebox(0,0)[b]{\smash{{\SetFigFont{9}{10.8}{\rmdefault}{\mddefault}{\updefault}$z$}}}}
\put(6706,-7576){\makebox(0,0)[b]{\smash{{\SetFigFont{8}{9.6}{\rmdefault}{\mddefault}{\updefault}$P$}}}}
\put(6976,-7936){\makebox(0,0)[lb]{\smash{{\SetFigFont{9}{10.8}{\rmdefault}{\mddefault}{\updefault}$a$}}}}
\put(6976,-8116){\makebox(0,0)[lb]{\smash{{\SetFigFont{9}{10.8}{\rmdefault}{\mddefault}{\updefault}$b$}}}}
\put(6706,-8476){\makebox(0,0)[b]{\smash{{\SetFigFont{9}{10.8}{\rmdefault}{\mddefault}{\updefault}$c$}}}}
\put(8416,-8836){\makebox(0,0)[b]{\smash{{\SetFigFont{8}{9.6}{\rmdefault}{\mddefault}{\updefault}$Q'$}}}}
\put(8416,-7576){\makebox(0,0)[b]{\smash{{\SetFigFont{8}{9.6}{\rmdefault}{\mddefault}{\updefault}$Q$}}}}
\put(8686,-7936){\makebox(0,0)[lb]{\smash{{\SetFigFont{9}{10.8}{\rmdefault}{\mddefault}{\updefault}$a$}}}}
\put(8686,-8116){\makebox(0,0)[lb]{\smash{{\SetFigFont{9}{10.8}{\rmdefault}{\mddefault}{\updefault}$b$}}}}
\put(8686,-8296){\makebox(0,0)[lb]{\smash{{\SetFigFont{9}{10.8}{\rmdefault}{\mddefault}{\updefault}$c$}}}}
\put(8686,-8476){\makebox(0,0)[lb]{\smash{{\SetFigFont{9}{10.8}{\rmdefault}{\mddefault}{\updefault}$z$}}}}
\put(8416,-11401){\makebox(0,0)[b]{\smash{{\SetFigFont{8}{9.6}{\rmdefault}{\mddefault}{\updefault}$Q'$}}}}
\put(8416,-10141){\makebox(0,0)[b]{\smash{{\SetFigFont{8}{9.6}{\rmdefault}{\mddefault}{\updefault}$Q$}}}}
\put(8686,-10501){\makebox(0,0)[lb]{\smash{{\SetFigFont{9}{10.8}{\rmdefault}{\mddefault}{\updefault}$a$}}}}
\put(8686,-10681){\makebox(0,0)[lb]{\smash{{\SetFigFont{9}{10.8}{\rmdefault}{\mddefault}{\updefault}$b$}}}}
\put(8686,-10861){\makebox(0,0)[lb]{\smash{{\SetFigFont{9}{10.8}{\rmdefault}{\mddefault}{\updefault}$c$}}}}
\put(8686,-11041){\makebox(0,0)[lb]{\smash{{\SetFigFont{9}{10.8}{\rmdefault}{\mddefault}{\updefault}$z$}}}}
\end{picture}%

%% file: tight-example.pdf_t
\begin{picture}(0,0)%
\includegraphics{tight-example.pdf}%
\end{picture}%
\setlength{\unitlength}{3522sp}%
\begingroup\makeatletter\ifx\SetFigFont\undefined%
\gdef\SetFigFont#1#2#3#4#5{%
  \reset@font\fontsize{#1}{#2pt}%
  \fontfamily{#3}\fontseries{#4}\fontshape{#5}%
  \selectfont}%
\fi\endgroup%
\begin{picture}(8371,9027)(600,-8792)
\put(3196,-1996){\makebox(0,0)[rb]{\smash{{\SetFigFont{6}{7.2}{\rmdefault}{\mddefault}{\updefault}{\color[rgb]{0,0,0}$11$}%
}}}}
\put(3061,-2986){\makebox(0,0)[b]{\smash{{\SetFigFont{6}{7.2}{\rmdefault}{\mddefault}{\updefault}{\color[rgb]{0,0,0}$15$}%
}}}}
\put(3241,-2986){\makebox(0,0)[b]{\smash{{\SetFigFont{6}{7.2}{\rmdefault}{\mddefault}{\updefault}{\color[rgb]{0,0,0}$16$}%
}}}}
\put(2206,-2716){\makebox(0,0)[lb]{\smash{{\SetFigFont{6}{7.2}{\rmdefault}{\mddefault}{\updefault}{\color[rgb]{0,0,0}$13$}%
}}}}
\put(2836,-2716){\makebox(0,0)[rb]{\smash{{\SetFigFont{6}{7.2}{\rmdefault}{\mddefault}{\updefault}{\color[rgb]{0,0,0}$14$}%
}}}}
\put(3196,-3526){\makebox(0,0)[rb]{\smash{{\SetFigFont{6}{7.2}{\rmdefault}{\mddefault}{\updefault}{\color[rgb]{0,0,0}$21$}%
}}}}
\put(3196,-3706){\makebox(0,0)[rb]{\smash{{\SetFigFont{6}{7.2}{\rmdefault}{\mddefault}{\updefault}{\color[rgb]{0,0,0}$22$}%
}}}}
\put(1846,-3526){\makebox(0,0)[lb]{\smash{{\SetFigFont{6}{7.2}{\rmdefault}{\mddefault}{\updefault}{\color[rgb]{0,0,0}$17$}%
}}}}
\put(1846,-3706){\makebox(0,0)[lb]{\smash{{\SetFigFont{6}{7.2}{\rmdefault}{\mddefault}{\updefault}{\color[rgb]{0,0,0}$18$}%
}}}}
\put(1846,-3886){\makebox(0,0)[lb]{\smash{{\SetFigFont{6}{7.2}{\rmdefault}{\mddefault}{\updefault}{\color[rgb]{0,0,0}$19$}%
}}}}
\put(1846,-4066){\makebox(0,0)[lb]{\smash{{\SetFigFont{6}{7.2}{\rmdefault}{\mddefault}{\updefault}{\color[rgb]{0,0,0}$20$}%
}}}}
\put(4231,-4246){\makebox(0,0)[lb]{\smash{{\SetFigFont{6}{7.2}{\rmdefault}{\mddefault}{\updefault}{\color[rgb]{0,0,0}$19$}%
}}}}
\put(4231,-4426){\makebox(0,0)[lb]{\smash{{\SetFigFont{6}{7.2}{\rmdefault}{\mddefault}{\updefault}{\color[rgb]{0,0,0}$20$}%
}}}}
\put(4006,-3706){\makebox(0,0)[b]{\smash{{\SetFigFont{6}{7.2}{\rmdefault}{\mddefault}{\updefault}{\color[rgb]{0,0,0}$17$}%
}}}}
\put(4006,-3976){\makebox(0,0)[b]{\smash{{\SetFigFont{6}{7.2}{\rmdefault}{\mddefault}{\updefault}{\color[rgb]{0,0,0}$18$}%
}}}}
\put(3061,-4876){\makebox(0,0)[b]{\smash{{\SetFigFont{6}{7.2}{\rmdefault}{\mddefault}{\updefault}{\color[rgb]{0,0,0}$26$}%
}}}}
\put(3241,-4876){\makebox(0,0)[b]{\smash{{\SetFigFont{6}{7.2}{\rmdefault}{\mddefault}{\updefault}{\color[rgb]{0,0,0}$27$}%
}}}}
\put(2836,-4606){\makebox(0,0)[rb]{\smash{{\SetFigFont{6}{7.2}{\rmdefault}{\mddefault}{\updefault}{\color[rgb]{0,0,0}$25$}%
}}}}
\put(4636,-4876){\makebox(0,0)[b]{\smash{{\SetFigFont{6}{7.2}{\rmdefault}{\mddefault}{\updefault}{\color[rgb]{0,0,0}$24$}%
}}}}
\put(5176,-4876){\makebox(0,0)[b]{\smash{{\SetFigFont{6}{7.2}{\rmdefault}{\mddefault}{\updefault}{\color[rgb]{0,0,0}$27$}%
}}}}
\put(4816,-4876){\makebox(0,0)[b]{\smash{{\SetFigFont{6}{7.2}{\rmdefault}{\mddefault}{\updefault}{\color[rgb]{0,0,0}$25$}%
}}}}
\put(4996,-4876){\makebox(0,0)[b]{\smash{{\SetFigFont{6}{7.2}{\rmdefault}{\mddefault}{\updefault}{\color[rgb]{0,0,0}$26$}%
}}}}
\put(1846,-5416){\makebox(0,0)[lb]{\smash{{\SetFigFont{6}{7.2}{\rmdefault}{\mddefault}{\updefault}{\color[rgb]{0,0,0}$28$}%
}}}}
\put(1846,-5596){\makebox(0,0)[lb]{\smash{{\SetFigFont{6}{7.2}{\rmdefault}{\mddefault}{\updefault}{\color[rgb]{0,0,0}$29$}%
}}}}
\put(1846,-5776){\makebox(0,0)[lb]{\smash{{\SetFigFont{6}{7.2}{\rmdefault}{\mddefault}{\updefault}{\color[rgb]{0,0,0}$30$}%
}}}}
\put(1846,-5956){\makebox(0,0)[lb]{\smash{{\SetFigFont{6}{7.2}{\rmdefault}{\mddefault}{\updefault}{\color[rgb]{0,0,0}$31$}%
}}}}
\put(4231,-6136){\makebox(0,0)[lb]{\smash{{\SetFigFont{6}{7.2}{\rmdefault}{\mddefault}{\updefault}{\color[rgb]{0,0,0}$30$}%
}}}}
\put(4231,-6316){\makebox(0,0)[lb]{\smash{{\SetFigFont{6}{7.2}{\rmdefault}{\mddefault}{\updefault}{\color[rgb]{0,0,0}$31$}%
}}}}
\put(4006,-5596){\makebox(0,0)[b]{\smash{{\SetFigFont{6}{7.2}{\rmdefault}{\mddefault}{\updefault}{\color[rgb]{0,0,0}$28$}%
}}}}
\put(4006,-5866){\makebox(0,0)[b]{\smash{{\SetFigFont{6}{7.2}{\rmdefault}{\mddefault}{\updefault}{\color[rgb]{0,0,0}$29$}%
}}}}
\put(2206,-6496){\makebox(0,0)[lb]{\smash{{\SetFigFont{6}{7.2}{\rmdefault}{\mddefault}{\updefault}{\color[rgb]{0,0,0}$35$}%
}}}}
\put(2836,-6496){\makebox(0,0)[rb]{\smash{{\SetFigFont{6}{7.2}{\rmdefault}{\mddefault}{\updefault}{\color[rgb]{0,0,0}$36$}%
}}}}
\put(3061,-6766){\makebox(0,0)[b]{\smash{{\SetFigFont{6}{7.2}{\rmdefault}{\mddefault}{\updefault}{\color[rgb]{0,0,0}$37$}%
}}}}
\put(3241,-6766){\makebox(0,0)[b]{\smash{{\SetFigFont{6}{7.2}{\rmdefault}{\mddefault}{\updefault}{\color[rgb]{0,0,0}$38$}%
}}}}
\put(4636,-6766){\makebox(0,0)[b]{\smash{{\SetFigFont{6}{7.2}{\rmdefault}{\mddefault}{\updefault}{\color[rgb]{0,0,0}$35$}%
}}}}
\put(5176,-6766){\makebox(0,0)[b]{\smash{{\SetFigFont{6}{7.2}{\rmdefault}{\mddefault}{\updefault}{\color[rgb]{0,0,0}$38$}%
}}}}
\put(4816,-6766){\makebox(0,0)[b]{\smash{{\SetFigFont{6}{7.2}{\rmdefault}{\mddefault}{\updefault}{\color[rgb]{0,0,0}$36$}%
}}}}
\put(4996,-6766){\makebox(0,0)[b]{\smash{{\SetFigFont{6}{7.2}{\rmdefault}{\mddefault}{\updefault}{\color[rgb]{0,0,0}$37$}%
}}}}
\put(3106,-7441){\makebox(0,0)[b]{\smash{{\SetFigFont{6}{7.2}{\rmdefault}{\mddefault}{\updefault}{\color[rgb]{0,0,0}$B$}%
}}}}
\put(4006,-8746){\makebox(0,0)[b]{\smash{{\SetFigFont{6}{7.2}{\rmdefault}{\mddefault}{\updefault}{\color[rgb]{0,0,0}$11k-11$}%
}}}}
\put(4231,-8476){\makebox(0,0)[lb]{\smash{{\SetFigFont{6}{7.2}{\rmdefault}{\mddefault}{\updefault}{\color[rgb]{0,0,0}$11k-10$}%
}}}}
\put(4231,-8296){\makebox(0,0)[lb]{\smash{{\SetFigFont{6}{7.2}{\rmdefault}{\mddefault}{\updefault}{\color[rgb]{0,0,0}$11k-9$}%
}}}}
\put(856,-2896){\makebox(0,0)[lb]{\smash{{\SetFigFont{6}{7.2}{\rmdefault}{\mddefault}{\updefault}{\color[rgb]{0,0,0}$11k-9$}%
}}}}
\put(856,-2716){\makebox(0,0)[lb]{\smash{{\SetFigFont{6}{7.2}{\rmdefault}{\mddefault}{\updefault}{\color[rgb]{0,0,0}$11k-10$}%
}}}}
\put(856,-2536){\makebox(0,0)[lb]{\smash{{\SetFigFont{6}{7.2}{\rmdefault}{\mddefault}{\updefault}{\color[rgb]{0,0,0}$11k-11$}%
}}}}
\put(856,-2356){\makebox(0,0)[lb]{\smash{{\SetFigFont{6}{7.2}{\rmdefault}{\mddefault}{\updefault}{\color[rgb]{0,0,0}$11-12$}%
}}}}
\put(3196,-8296){\makebox(0,0)[rb]{\smash{{\SetFigFont{6}{7.2}{\rmdefault}{\mddefault}{\updefault}{\color[rgb]{0,0,0}$11k-13$}%
}}}}
\put(3196,-8476){\makebox(0,0)[rb]{\smash{{\SetFigFont{6}{7.2}{\rmdefault}{\mddefault}{\updefault}{\color[rgb]{0,0,0}$11k-14$}%
}}}}
\put(3421,-8746){\makebox(0,0)[b]{\smash{{\SetFigFont{6}{7.2}{\rmdefault}{\mddefault}{\updefault}{\color[rgb]{0,0,0}$11k-15$}%
}}}}
\put(6571,-2896){\makebox(0,0)[rb]{\smash{{\SetFigFont{6}{7.2}{\rmdefault}{\mddefault}{\updefault}{\color[rgb]{0,0,0}$11k-13$}%
}}}}
\put(6571,-2716){\makebox(0,0)[rb]{\smash{{\SetFigFont{6}{7.2}{\rmdefault}{\mddefault}{\updefault}{\color[rgb]{0,0,0}$11k-14$}%
}}}}
\put(6571,-2536){\makebox(0,0)[rb]{\smash{{\SetFigFont{6}{7.2}{\rmdefault}{\mddefault}{\updefault}{\color[rgb]{0,0,0}$11k-15$}%
}}}}
\put(6571,-2356){\makebox(0,0)[rb]{\smash{{\SetFigFont{6}{7.2}{\rmdefault}{\mddefault}{\updefault}{\color[rgb]{0,0,0}$11k-16$}%
}}}}
\put(7786,-4471){\makebox(0,0)[b]{\smash{{\SetFigFont{6}{7.2}{\rmdefault}{\mddefault}{\updefault}{\color[rgb]{0,0,0}$4$}%
}}}}
\put(7651,-4111){\makebox(0,0)[b]{\smash{{\SetFigFont{6}{7.2}{\rmdefault}{\mddefault}{\updefault}{\color[rgb]{0,0,0}$A$}%
}}}}
\put(7966,-4111){\makebox(0,0)[b]{\smash{{\SetFigFont{6}{7.2}{\rmdefault}{\mddefault}{\updefault}{\color[rgb]{0,0,0}$B$}%
}}}}
\put(2701, 74){\makebox(0,0)[b]{\smash{{\SetFigFont{8}{9.6}{\rmdefault}{\mddefault}{\updefault}{\color[rgb]{0,0,0}$T_k$}%
}}}}
\put(1621,-466){\makebox(0,0)[b]{\smash{{\SetFigFont{6}{7.2}{\rmdefault}{\mddefault}{\updefault}{\color[rgb]{0,0,0}$11k-16$}%
}}}}
\put(4996,-7441){\makebox(0,0)[b]{\smash{{\SetFigFont{6}{7.2}{\rmdefault}{\mddefault}{\updefault}{\color[rgb]{0,0,0}$B$}%
}}}}
\put(2089,-7441){\makebox(0,0)[b]{\smash{{\SetFigFont{6}{7.2}{\rmdefault}{\mddefault}{\updefault}{\color[rgb]{0,0,0}$A$}%
}}}}
\put(4568,-7446){\makebox(0,0)[b]{\smash{{\SetFigFont{6}{7.2}{\rmdefault}{\mddefault}{\updefault}{\color[rgb]{0,0,0}$A$}%
}}}}
\put(2836,-826){\makebox(0,0)[rb]{\smash{{\SetFigFont{6}{7.2}{\rmdefault}{\mddefault}{\updefault}{\color[rgb]{0,0,0}$2$}%
}}}}
\put(4636,-1096){\makebox(0,0)[b]{\smash{{\SetFigFont{6}{7.2}{\rmdefault}{\mddefault}{\updefault}{\color[rgb]{0,0,0}$1$}%
}}}}
\put(4231,-2356){\makebox(0,0)[lb]{\smash{{\SetFigFont{6}{7.2}{\rmdefault}{\mddefault}{\updefault}{\color[rgb]{0,0,0}$7$}%
}}}}
\put(4231,-2536){\makebox(0,0)[lb]{\smash{{\SetFigFont{6}{7.2}{\rmdefault}{\mddefault}{\updefault}{\color[rgb]{0,0,0}$8$}%
}}}}
\put(5581,-2356){\makebox(0,0)[rb]{\smash{{\SetFigFont{6}{7.2}{\rmdefault}{\mddefault}{\updefault}{\color[rgb]{0,0,0}$11$}%
}}}}
\put(5581,-2536){\makebox(0,0)[rb]{\smash{{\SetFigFont{6}{7.2}{\rmdefault}{\mddefault}{\updefault}{\color[rgb]{0,0,0}$12$}%
}}}}
\put(3196,-3886){\makebox(0,0)[rb]{\smash{{\SetFigFont{6}{7.2}{\rmdefault}{\mddefault}{\updefault}{\color[rgb]{0,0,0}$23$}%
}}}}
\put(5581,-3526){\makebox(0,0)[rb]{\smash{{\SetFigFont{6}{7.2}{\rmdefault}{\mddefault}{\updefault}{\color[rgb]{0,0,0}$21$}%
}}}}
\put(5581,-3706){\makebox(0,0)[rb]{\smash{{\SetFigFont{6}{7.2}{\rmdefault}{\mddefault}{\updefault}{\color[rgb]{0,0,0}$22$}%
}}}}
\put(5581,-3886){\makebox(0,0)[rb]{\smash{{\SetFigFont{6}{7.2}{\rmdefault}{\mddefault}{\updefault}{\color[rgb]{0,0,0}$23$}%
}}}}
\put(2206,-4606){\makebox(0,0)[lb]{\smash{{\SetFigFont{6}{7.2}{\rmdefault}{\mddefault}{\updefault}{\color[rgb]{0,0,0}$24$}%
}}}}
\put(5806,-466){\makebox(0,0)[b]{\smash{{\SetFigFont{6}{7.2}{\rmdefault}{\mddefault}{\updefault}{\color[rgb]{0,0,0}$11k-12$}%
}}}}
\put(4006,-2086){\makebox(0,0)[b]{\smash{{\SetFigFont{6}{7.2}{\rmdefault}{\mddefault}{\updefault}{\color[rgb]{0,0,0}$6$}%
}}}}
\put(3196,-5416){\makebox(0,0)[rb]{\smash{{\SetFigFont{6}{7.2}{\rmdefault}{\mddefault}{\updefault}{\color[rgb]{0,0,0}$32$}%
}}}}
\put(3196,-5596){\makebox(0,0)[rb]{\smash{{\SetFigFont{6}{7.2}{\rmdefault}{\mddefault}{\updefault}{\color[rgb]{0,0,0}$33$}%
}}}}
\put(3196,-5776){\makebox(0,0)[rb]{\smash{{\SetFigFont{6}{7.2}{\rmdefault}{\mddefault}{\updefault}{\color[rgb]{0,0,0}$34$}%
}}}}
\put(5581,-5416){\makebox(0,0)[rb]{\smash{{\SetFigFont{6}{7.2}{\rmdefault}{\mddefault}{\updefault}{\color[rgb]{0,0,0}$32$}%
}}}}
\put(5581,-5596){\makebox(0,0)[rb]{\smash{{\SetFigFont{6}{7.2}{\rmdefault}{\mddefault}{\updefault}{\color[rgb]{0,0,0}$33$}%
}}}}
\put(5581,-5776){\makebox(0,0)[rb]{\smash{{\SetFigFont{6}{7.2}{\rmdefault}{\mddefault}{\updefault}{\color[rgb]{0,0,0}$34$}%
}}}}
\put(8956,-1321){\makebox(0,0)[b]{\smash{{\SetFigFont{8}{9.6}{\rmdefault}{\mddefault}{\updefault}{\color[rgb]{0,0,0}$G_4$}%
}}}}
\put(7786,-3076){\makebox(0,0)[b]{\smash{{\SetFigFont{6}{7.2}{\rmdefault}{\mddefault}{\updefault}{\color[rgb]{0,0,0}$4$}%
}}}}
\put(7786,-2761){\makebox(0,0)[b]{\smash{{\SetFigFont{6}{7.2}{\rmdefault}{\mddefault}{\updefault}{\color[rgb]{0,0,0}$4$}%
}}}}
\put(7471,-3166){\makebox(0,0)[lb]{\smash{{\SetFigFont{6}{7.2}{\rmdefault}{\mddefault}{\updefault}{\color[rgb]{0,0,0}$4$}%
}}}}
\put(8146,-2851){\makebox(0,0)[rb]{\smash{{\SetFigFont{6}{7.2}{\rmdefault}{\mddefault}{\updefault}{\color[rgb]{0,0,0}$4$}%
}}}}
\put(8146,-3166){\makebox(0,0)[rb]{\smash{{\SetFigFont{6}{7.2}{\rmdefault}{\mddefault}{\updefault}{\color[rgb]{0,0,0}$3$}%
}}}}
\put(7471,-2851){\makebox(0,0)[lb]{\smash{{\SetFigFont{6}{7.2}{\rmdefault}{\mddefault}{\updefault}{\color[rgb]{0,0,0}$4$}%
}}}}
\put(7786,-3391){\makebox(0,0)[b]{\smash{{\SetFigFont{6}{7.2}{\rmdefault}{\mddefault}{\updefault}{\color[rgb]{0,0,0}$4$}%
}}}}
\put(7786,-3706){\makebox(0,0)[b]{\smash{{\SetFigFont{6}{7.2}{\rmdefault}{\mddefault}{\updefault}{\color[rgb]{0,0,0}$4$}%
}}}}
\put(7471,-3436){\makebox(0,0)[lb]{\smash{{\SetFigFont{6}{7.2}{\rmdefault}{\mddefault}{\updefault}{\color[rgb]{0,0,0}$4$}%
}}}}
\put(8146,-3436){\makebox(0,0)[rb]{\smash{{\SetFigFont{6}{7.2}{\rmdefault}{\mddefault}{\updefault}{\color[rgb]{0,0,0}$3$}%
}}}}
\put(7201,-2581){\makebox(0,0)[rb]{\smash{{\SetFigFont{6}{7.2}{\rmdefault}{\mddefault}{\updefault}{\color[rgb]{0,0,0}$4$}%
}}}}
\put(8416,-2581){\makebox(0,0)[lb]{\smash{{\SetFigFont{6}{7.2}{\rmdefault}{\mddefault}{\updefault}{\color[rgb]{0,0,0}$4$}%
}}}}
\put(8956,-3256){\makebox(0,0)[b]{\smash{{\SetFigFont{8}{9.6}{\rmdefault}{\mddefault}{\updefault}{\color[rgb]{0,0,0}$G_k$}%
}}}}
\put(4726, 74){\makebox(0,0)[b]{\smash{{\SetFigFont{8}{9.6}{\rmdefault}{\mddefault}{\updefault}{\color[rgb]{0,0,0}$T_k'$}%
}}}}
\put(7786,-1636){\makebox(0,0)[b]{\smash{{\SetFigFont{6}{7.2}{\rmdefault}{\mddefault}{\updefault}{\color[rgb]{0,0,0}$4$}%
}}}}
\put(7786,-1321){\makebox(0,0)[b]{\smash{{\SetFigFont{6}{7.2}{\rmdefault}{\mddefault}{\updefault}{\color[rgb]{0,0,0}$4$}%
}}}}
\put(7471,-1726){\makebox(0,0)[lb]{\smash{{\SetFigFont{6}{7.2}{\rmdefault}{\mddefault}{\updefault}{\color[rgb]{0,0,0}$4$}%
}}}}
\put(8146,-1411){\makebox(0,0)[rb]{\smash{{\SetFigFont{6}{7.2}{\rmdefault}{\mddefault}{\updefault}{\color[rgb]{0,0,0}$4$}%
}}}}
\put(8146,-1726){\makebox(0,0)[rb]{\smash{{\SetFigFont{6}{7.2}{\rmdefault}{\mddefault}{\updefault}{\color[rgb]{0,0,0}$3$}%
}}}}
\put(7471,-1411){\makebox(0,0)[lb]{\smash{{\SetFigFont{6}{7.2}{\rmdefault}{\mddefault}{\updefault}{\color[rgb]{0,0,0}$4$}%
}}}}
\put(7201,-1141){\makebox(0,0)[rb]{\smash{{\SetFigFont{6}{7.2}{\rmdefault}{\mddefault}{\updefault}{\color[rgb]{0,0,0}$4$}%
}}}}
\put(8416,-1141){\makebox(0,0)[lb]{\smash{{\SetFigFont{6}{7.2}{\rmdefault}{\mddefault}{\updefault}{\color[rgb]{0,0,0}$4$}%
}}}}
\put(7786,-1951){\makebox(0,0)[b]{\smash{{\SetFigFont{6}{7.2}{\rmdefault}{\mddefault}{\updefault}{\color[rgb]{0,0,0}$4$}%
}}}}
\put(4816,-1096){\makebox(0,0)[b]{\smash{{\SetFigFont{6}{7.2}{\rmdefault}{\mddefault}{\updefault}{\color[rgb]{0,0,0}$2$}%
}}}}
\put(2206,-826){\makebox(0,0)[lb]{\smash{{\SetFigFont{6}{7.2}{\rmdefault}{\mddefault}{\updefault}{\color[rgb]{0,0,0}$1$}%
}}}}
\put(4996,-1096){\makebox(0,0)[b]{\smash{{\SetFigFont{6}{7.2}{\rmdefault}{\mddefault}{\updefault}{\color[rgb]{0,0,0}$3$}%
}}}}
\put(5176,-1096){\makebox(0,0)[b]{\smash{{\SetFigFont{6}{7.2}{\rmdefault}{\mddefault}{\updefault}{\color[rgb]{0,0,0}$4$}%
}}}}
\put(3241,-1096){\makebox(0,0)[b]{\smash{{\SetFigFont{6}{7.2}{\rmdefault}{\mddefault}{\updefault}{\color[rgb]{0,0,0}$4$}%
}}}}
\put(3061,-1096){\makebox(0,0)[b]{\smash{{\SetFigFont{6}{7.2}{\rmdefault}{\mddefault}{\updefault}{\color[rgb]{0,0,0}$3$}%
}}}}
\put(4006,-1816){\makebox(0,0)[b]{\smash{{\SetFigFont{6}{7.2}{\rmdefault}{\mddefault}{\updefault}{\color[rgb]{0,0,0}$5$}%
}}}}
\put(5806,-1816){\makebox(0,0)[b]{\smash{{\SetFigFont{6}{7.2}{\rmdefault}{\mddefault}{\updefault}{\color[rgb]{0,0,0}$9$}%
}}}}
\put(5806,-2086){\makebox(0,0)[b]{\smash{{\SetFigFont{6}{7.2}{\rmdefault}{\mddefault}{\updefault}{\color[rgb]{0,0,0}$10$}%
}}}}
\put(4636,-2986){\makebox(0,0)[b]{\smash{{\SetFigFont{6}{7.2}{\rmdefault}{\mddefault}{\updefault}{\color[rgb]{0,0,0}$13$}%
}}}}
\put(4816,-2986){\makebox(0,0)[b]{\smash{{\SetFigFont{6}{7.2}{\rmdefault}{\mddefault}{\updefault}{\color[rgb]{0,0,0}$14$}%
}}}}
\put(4996,-2986){\makebox(0,0)[b]{\smash{{\SetFigFont{6}{7.2}{\rmdefault}{\mddefault}{\updefault}{\color[rgb]{0,0,0}$15$}%
}}}}
\put(5176,-2986){\makebox(0,0)[b]{\smash{{\SetFigFont{6}{7.2}{\rmdefault}{\mddefault}{\updefault}{\color[rgb]{0,0,0}$16$}%
}}}}
\put(1846,-1636){\makebox(0,0)[lb]{\smash{{\SetFigFont{6}{7.2}{\rmdefault}{\mddefault}{\updefault}{\color[rgb]{0,0,0}$5$}%
}}}}
\put(1846,-2176){\makebox(0,0)[lb]{\smash{{\SetFigFont{6}{7.2}{\rmdefault}{\mddefault}{\updefault}{\color[rgb]{0,0,0}$8$}%
}}}}
\put(1846,-1816){\makebox(0,0)[lb]{\smash{{\SetFigFont{6}{7.2}{\rmdefault}{\mddefault}{\updefault}{\color[rgb]{0,0,0}$6$}%
}}}}
\put(1846,-1996){\makebox(0,0)[lb]{\smash{{\SetFigFont{6}{7.2}{\rmdefault}{\mddefault}{\updefault}{\color[rgb]{0,0,0}$7$}%
}}}}
\put(3196,-1636){\makebox(0,0)[rb]{\smash{{\SetFigFont{6}{7.2}{\rmdefault}{\mddefault}{\updefault}{\color[rgb]{0,0,0}$9$}%
}}}}
\put(3196,-2176){\makebox(0,0)[rb]{\smash{{\SetFigFont{6}{7.2}{\rmdefault}{\mddefault}{\updefault}{\color[rgb]{0,0,0}$12$}%
}}}}
\put(3196,-1816){\makebox(0,0)[rb]{\smash{{\SetFigFont{6}{7.2}{\rmdefault}{\mddefault}{\updefault}{\color[rgb]{0,0,0}$10$}%
}}}}
\end{picture}%

%% file: 3-chain-proof.pdf_t
\begin{picture}(0,0)%
\includegraphics{3-chain-proof.pdf}%
\end{picture}%
\setlength{\unitlength}{3522sp}%
\begingroup\makeatletter\ifx\SetFigFont\undefined%
\gdef\SetFigFont#1#2#3#4#5{%
  \reset@font\fontsize{#1}{#2pt}%
  \fontfamily{#3}\fontseries{#4}\fontshape{#5}%
  \selectfont}%
\fi\endgroup%
\begin{picture}(5847,3553)(2224,-6020)
\put(2431,-4471){\makebox(0,0)[b]{\smash{{\SetFigFont{9}{10.8}{\rmdefault}{\mddefault}{\updefault}(ii)}}}}
\put(2791,-2851){\makebox(0,0)[lb]{\smash{{\SetFigFont{8}{9.6}{\rmdefault}{\mddefault}{\updefault}$L_{T}(C)$}}}}
\put(5356,-2851){\makebox(0,0)[rb]{\smash{{\SetFigFont{8}{9.6}{\rmdefault}{\mddefault}{\updefault}$R_{T}(C)$}}}}
\put(5491,-2851){\makebox(0,0)[lb]{\smash{{\SetFigFont{8}{9.6}{\rmdefault}{\mddefault}{\updefault}$L_{T'}(C)$}}}}
\put(8056,-2851){\makebox(0,0)[rb]{\smash{{\SetFigFont{8}{9.6}{\rmdefault}{\mddefault}{\updefault}$R_{T'}(C)$}}}}
\put(4051,-2626){\makebox(0,0)[b]{\smash{{\SetFigFont{10}{12.0}{\rmdefault}{\mddefault}{\updefault}$T$}}}}
\put(6751,-2626){\makebox(0,0)[b]{\smash{{\SetFigFont{10}{12.0}{\rmdefault}{\mddefault}{\updefault}$T'$}}}}
\put(2386,-5596){\makebox(0,0)[b]{\smash{{\SetFigFont{9}{10.8}{\rmdefault}{\mddefault}{\updefault}(iii)}}}}
\put(2386,-3346){\makebox(0,0)[b]{\smash{{\SetFigFont{9}{10.8}{\rmdefault}{\mddefault}{\updefault}(i)}}}}
\put(3691,-3706){\makebox(0,0)[b]{\smash{{\SetFigFont{9}{10.8}{\rmdefault}{\mddefault}{\updefault}{\color[rgb]{0,0,0}$a$}%
}}}}
\put(4051,-3706){\makebox(0,0)[b]{\smash{{\SetFigFont{9}{10.8}{\rmdefault}{\mddefault}{\updefault}{\color[rgb]{0,0,0}$b$}%
}}}}
\put(4411,-3706){\makebox(0,0)[b]{\smash{{\SetFigFont{9}{10.8}{\rmdefault}{\mddefault}{\updefault}{\color[rgb]{0,0,0}$c$}%
}}}}
\put(6391,-3706){\makebox(0,0)[b]{\smash{{\SetFigFont{9}{10.8}{\rmdefault}{\mddefault}{\updefault}{\color[rgb]{0,0,0}$c$}%
}}}}
\put(6751,-3706){\makebox(0,0)[b]{\smash{{\SetFigFont{9}{10.8}{\rmdefault}{\mddefault}{\updefault}{\color[rgb]{0,0,0}$b$}%
}}}}
\put(7111,-3706){\makebox(0,0)[b]{\smash{{\SetFigFont{9}{10.8}{\rmdefault}{\mddefault}{\updefault}{\color[rgb]{0,0,0}$a$}%
}}}}
\put(7111,-4831){\makebox(0,0)[b]{\smash{{\SetFigFont{9}{10.8}{\rmdefault}{\mddefault}{\updefault}{\color[rgb]{0,0,0}$a$}%
}}}}
\put(6751,-4831){\makebox(0,0)[b]{\smash{{\SetFigFont{9}{10.8}{\rmdefault}{\mddefault}{\updefault}{\color[rgb]{0,0,0}$b$}%
}}}}
\put(6391,-4831){\makebox(0,0)[b]{\smash{{\SetFigFont{9}{10.8}{\rmdefault}{\mddefault}{\updefault}{\color[rgb]{0,0,0}$c$}%
}}}}
\put(6391,-5956){\makebox(0,0)[b]{\smash{{\SetFigFont{9}{10.8}{\rmdefault}{\mddefault}{\updefault}{\color[rgb]{0,0,0}$c$}%
}}}}
\put(6751,-5956){\makebox(0,0)[b]{\smash{{\SetFigFont{9}{10.8}{\rmdefault}{\mddefault}{\updefault}{\color[rgb]{0,0,0}$b$}%
}}}}
\put(7111,-5956){\makebox(0,0)[b]{\smash{{\SetFigFont{9}{10.8}{\rmdefault}{\mddefault}{\updefault}{\color[rgb]{0,0,0}$a$}%
}}}}
\put(4411,-5956){\makebox(0,0)[b]{\smash{{\SetFigFont{9}{10.8}{\rmdefault}{\mddefault}{\updefault}{\color[rgb]{0,0,0}$c$}%
}}}}
\put(4051,-5956){\makebox(0,0)[b]{\smash{{\SetFigFont{9}{10.8}{\rmdefault}{\mddefault}{\updefault}{\color[rgb]{0,0,0}$b$}%
}}}}
\put(3691,-5956){\makebox(0,0)[b]{\smash{{\SetFigFont{9}{10.8}{\rmdefault}{\mddefault}{\updefault}{\color[rgb]{0,0,0}$a$}%
}}}}
\put(3691,-4831){\makebox(0,0)[b]{\smash{{\SetFigFont{9}{10.8}{\rmdefault}{\mddefault}{\updefault}{\color[rgb]{0,0,0}$a$}%
}}}}
\put(4051,-4831){\makebox(0,0)[b]{\smash{{\SetFigFont{9}{10.8}{\rmdefault}{\mddefault}{\updefault}{\color[rgb]{0,0,0}$b$}%
}}}}
\put(4411,-4831){\makebox(0,0)[b]{\smash{{\SetFigFont{9}{10.8}{\rmdefault}{\mddefault}{\updefault}{\color[rgb]{0,0,0}$c$}%
}}}}
\end{picture}%

%% file: 2-chain-proof-V2.pdf_t
\begin{picture}(0,0)%
\includegraphics{2-chain-proof-V2.pdf}%
\end{picture}%
\setlength{\unitlength}{3522sp}%
\begingroup\makeatletter\ifx\SetFigFont\undefined%
\gdef\SetFigFont#1#2#3#4#5{%
  \reset@font\fontsize{#1}{#2pt}%
  \fontfamily{#3}\fontseries{#4}\fontshape{#5}%
  \selectfont}%
\fi\endgroup%
\begin{picture}(5014,2608)(2247,-4940)
\put(3871,-2491){\makebox(0,0)[b]{\smash{{\SetFigFont{10}{12.0}{\rmdefault}{\mddefault}{\updefault}$T$}}}}
\put(6706,-2491){\makebox(0,0)[b]{\smash{{\SetFigFont{10}{12.0}{\rmdefault}{\mddefault}{\updefault}$T'$}}}}
\put(2791,-2851){\makebox(0,0)[lb]{\smash{{\SetFigFont{8}{9.6}{\rmdefault}{\mddefault}{\updefault}$L_{T}(C)$}}}}
\put(4951,-2851){\makebox(0,0)[rb]{\smash{{\SetFigFont{8}{9.6}{\rmdefault}{\mddefault}{\updefault}$R_{T}(C)$}}}}
\put(2431,-3391){\makebox(0,0)[b]{\smash{{\SetFigFont{9}{10.8}{\rmdefault}{\mddefault}{\updefault}(i)}}}}
\put(2386,-4516){\makebox(0,0)[b]{\smash{{\SetFigFont{9}{10.8}{\rmdefault}{\mddefault}{\updefault}(ii)}}}}
\put(3691,-3751){\makebox(0,0)[b]{\smash{{\SetFigFont{9}{10.8}{\rmdefault}{\mddefault}{\updefault}{\color[rgb]{0,0,0}$a$}%
}}}}
\put(4051,-3751){\makebox(0,0)[b]{\smash{{\SetFigFont{9}{10.8}{\rmdefault}{\mddefault}{\updefault}{\color[rgb]{0,0,0}$b$}%
}}}}
\put(3691,-4876){\makebox(0,0)[b]{\smash{{\SetFigFont{9}{10.8}{\rmdefault}{\mddefault}{\updefault}{\color[rgb]{0,0,0}$a$}%
}}}}
\put(4051,-4876){\makebox(0,0)[b]{\smash{{\SetFigFont{9}{10.8}{\rmdefault}{\mddefault}{\updefault}{\color[rgb]{0,0,0}$b$}%
}}}}
\put(6211,-4516){\makebox(0,0)[b]{\smash{{\SetFigFont{7}{8.4}{\rmdefault}{\mddefault}{\updefault}{\color[rgb]{0,0,0}$Q$}%
}}}}
\put(4681,-4471){\makebox(0,0)[lb]{\smash{{\SetFigFont{7}{8.4}{\rmdefault}{\mddefault}{\updefault}{\color[rgb]{0,0,0}$P'$}%
}}}}
\put(3016,-4471){\makebox(0,0)[rb]{\smash{{\SetFigFont{7}{8.4}{\rmdefault}{\mddefault}{\updefault}{\color[rgb]{0,0,0}$P$}%
}}}}
\put(7246,-3121){\makebox(0,0)[lb]{\smash{{\SetFigFont{9}{10.8}{\rmdefault}{\mddefault}{\updefault}{\color[rgb]{0,0,0}$a$}%
}}}}
\put(7246,-3706){\makebox(0,0)[lb]{\smash{{\SetFigFont{9}{10.8}{\rmdefault}{\mddefault}{\updefault}{\color[rgb]{0,0,0}$b$}%
}}}}
\put(7246,-4246){\makebox(0,0)[lb]{\smash{{\SetFigFont{9}{10.8}{\rmdefault}{\mddefault}{\updefault}{\color[rgb]{0,0,0}$a$}%
}}}}
\put(7246,-4831){\makebox(0,0)[lb]{\smash{{\SetFigFont{9}{10.8}{\rmdefault}{\mddefault}{\updefault}{\color[rgb]{0,0,0}$b$}%
}}}}
\end{picture}%

%% file: Case-1-2-1.pdf_t
\begin{picture}(0,0)%
\includegraphics{Case-1-2-1.pdf}%
\end{picture}%
\setlength{\unitlength}{3522sp}%
\begingroup\makeatletter\ifx\SetFigFont\undefined%
\gdef\SetFigFont#1#2#3#4#5{%
  \reset@font\fontsize{#1}{#2pt}%
  \fontfamily{#3}\fontseries{#4}\fontshape{#5}%
  \selectfont}%
\fi\endgroup%
\begin{picture}(8486,2499)(2652,-2413)
\put(10711,-1366){\makebox(0,0)[b]{\smash{{\SetFigFont{9}{10.8}{\rmdefault}{\mddefault}{\updefault}{\color[rgb]{0,0,0}$C$}%
}}}}
\put(8551,-1636){\makebox(0,0)[lb]{\smash{{\SetFigFont{9}{10.8}{\rmdefault}{\mddefault}{\updefault}{\color[rgb]{0,0,0}$2$}%
}}}}
\put(8551,-1861){\makebox(0,0)[lb]{\smash{{\SetFigFont{9}{10.8}{\rmdefault}{\mddefault}{\updefault}{\color[rgb]{0,0,0}$n$}%
}}}}
\put(8461,-556){\makebox(0,0)[b]{\smash{{\SetFigFont{9}{10.8}{\rmdefault}{\mddefault}{\updefault}{\color[rgb]{0,0,0}$B_1$}%
}}}}
\put(8911,-961){\makebox(0,0)[b]{\smash{{\SetFigFont{9}{10.8}{\rmdefault}{\mddefault}{\updefault}{\color[rgb]{0,0,0}$1$}%
}}}}
\put(9091,-961){\makebox(0,0)[b]{\smash{{\SetFigFont{9}{10.8}{\rmdefault}{\mddefault}{\updefault}{\color[rgb]{0,0,0}$3$}%
}}}}
\put(9451,-961){\makebox(0,0)[b]{\smash{{\SetFigFont{9}{10.8}{\rmdefault}{\mddefault}{\updefault}{\color[rgb]{0,0,0}$n-1$}%
}}}}
\put(10891,-961){\makebox(0,0)[b]{\smash{{\SetFigFont{9}{10.8}{\rmdefault}{\mddefault}{\updefault}{\color[rgb]{0,0,0}$n$}%
}}}}
\put(10261,-961){\makebox(0,0)[b]{\smash{{\SetFigFont{9}{10.8}{\rmdefault}{\mddefault}{\updefault}{\color[rgb]{0,0,0}$2$}%
}}}}
\put(10081,-961){\makebox(0,0)[b]{\smash{{\SetFigFont{9}{10.8}{\rmdefault}{\mddefault}{\updefault}{\color[rgb]{0,0,0}$1$}%
}}}}
\put(10441,-961){\makebox(0,0)[b]{\smash{{\SetFigFont{9}{10.8}{\rmdefault}{\mddefault}{\updefault}{\color[rgb]{0,0,0}$3$}%
}}}}
\put(10576,-151){\makebox(0,0)[b]{\smash{{\SetFigFont{10}{12.0}{\rmdefault}{\mddefault}{\updefault}$F^*$}}}}
\put(10486,-556){\makebox(0,0)[b]{\smash{{\SetFigFont{9}{10.8}{\rmdefault}{\mddefault}{\updefault}{\color[rgb]{0,0,0}$B_1\cup C$}%
}}}}
\put(8371,-961){\makebox(0,0)[b]{\smash{{\SetFigFont{9}{10.8}{\rmdefault}{\mddefault}{\updefault}{\color[rgb]{0,0,0}$2$}%
}}}}
\put(8551,-961){\makebox(0,0)[b]{\smash{{\SetFigFont{9}{10.8}{\rmdefault}{\mddefault}{\updefault}{\color[rgb]{0,0,0}$n$}%
}}}}
\put(8551,-961){\makebox(0,0)[b]{\smash{{\SetFigFont{9}{10.8}{\rmdefault}{\mddefault}{\updefault}{\color[rgb]{0,0,0}$n$}%
}}}}
\put(9001,-1951){\makebox(0,0)[b]{\smash{{\SetFigFont{9}{10.8}{\rmdefault}{\mddefault}{\updefault}{\color[rgb]{0,0,0}$1$}%
}}}}
\put(9181,-1951){\makebox(0,0)[b]{\smash{{\SetFigFont{9}{10.8}{\rmdefault}{\mddefault}{\updefault}{\color[rgb]{0,0,0}$3$}%
}}}}
\put(9541,-1951){\makebox(0,0)[b]{\smash{{\SetFigFont{9}{10.8}{\rmdefault}{\mddefault}{\updefault}{\color[rgb]{0,0,0}$n-1$}%
}}}}
\put(4096,-151){\makebox(0,0)[b]{\smash{{\SetFigFont{10}{12.0}{\rmdefault}{\mddefault}{\updefault}$T$}}}}
\put(6796,-151){\makebox(0,0)[b]{\smash{{\SetFigFont{10}{12.0}{\rmdefault}{\mddefault}{\updefault}$T'$}}}}
\put(2791,-736){\makebox(0,0)[b]{\smash{{\SetFigFont{9}{10.8}{\rmdefault}{\mddefault}{\updefault}(i)}}}}
\put(7381,-1096){\makebox(0,0)[b]{\smash{{\SetFigFont{9}{10.8}{\rmdefault}{\mddefault}{\updefault}{\color[rgb]{0,0,0}$n$}%
}}}}
\put(6571,-1096){\makebox(0,0)[b]{\smash{{\SetFigFont{9}{10.8}{\rmdefault}{\mddefault}{\updefault}{\color[rgb]{0,0,0}$3$}%
}}}}
\put(6391,-1096){\makebox(0,0)[b]{\smash{{\SetFigFont{9}{10.8}{\rmdefault}{\mddefault}{\updefault}{\color[rgb]{0,0,0}$2$}%
}}}}
\put(6211,-1096){\makebox(0,0)[b]{\smash{{\SetFigFont{9}{10.8}{\rmdefault}{\mddefault}{\updefault}{\color[rgb]{0,0,0}$1$}%
}}}}
\put(8911,-151){\makebox(0,0)[b]{\smash{{\SetFigFont{10}{12.0}{\rmdefault}{\mddefault}{\updefault}$F$}}}}
\put(2791,-1726){\makebox(0,0)[b]{\smash{{\SetFigFont{9}{10.8}{\rmdefault}{\mddefault}{\updefault}(ii)}}}}
\put(4771,-2086){\makebox(0,0)[b]{\smash{{\SetFigFont{9}{10.8}{\rmdefault}{\mddefault}{\updefault}{\color[rgb]{0,0,0}$n$}%
}}}}
\put(3961,-2086){\makebox(0,0)[b]{\smash{{\SetFigFont{9}{10.8}{\rmdefault}{\mddefault}{\updefault}{\color[rgb]{0,0,0}$3$}%
}}}}
\put(3781,-2086){\makebox(0,0)[b]{\smash{{\SetFigFont{9}{10.8}{\rmdefault}{\mddefault}{\updefault}{\color[rgb]{0,0,0}$2$}%
}}}}
\put(3601,-2086){\makebox(0,0)[b]{\smash{{\SetFigFont{9}{10.8}{\rmdefault}{\mddefault}{\updefault}{\color[rgb]{0,0,0}$1$}%
}}}}
\put(7381,-2086){\makebox(0,0)[b]{\smash{{\SetFigFont{9}{10.8}{\rmdefault}{\mddefault}{\updefault}{\color[rgb]{0,0,0}$n$}%
}}}}
\put(6571,-2086){\makebox(0,0)[b]{\smash{{\SetFigFont{9}{10.8}{\rmdefault}{\mddefault}{\updefault}{\color[rgb]{0,0,0}$3$}%
}}}}
\put(6391,-2086){\makebox(0,0)[b]{\smash{{\SetFigFont{9}{10.8}{\rmdefault}{\mddefault}{\updefault}{\color[rgb]{0,0,0}$2$}%
}}}}
\put(6211,-2086){\makebox(0,0)[b]{\smash{{\SetFigFont{9}{10.8}{\rmdefault}{\mddefault}{\updefault}{\color[rgb]{0,0,0}$1$}%
}}}}
\put(8461,-1411){\makebox(0,0)[b]{\smash{{\SetFigFont{9}{10.8}{\rmdefault}{\mddefault}{\updefault}{\color[rgb]{0,0,0}$B_1$}%
}}}}
\put(10081,-1366){\makebox(0,0)[b]{\smash{{\SetFigFont{9}{10.8}{\rmdefault}{\mddefault}{\updefault}{\color[rgb]{0,0,0}$B_1\backslash C$}%
}}}}
\put(4771,-1096){\makebox(0,0)[b]{\smash{{\SetFigFont{9}{10.8}{\rmdefault}{\mddefault}{\updefault}{\color[rgb]{0,0,0}$n$}%
}}}}
\put(3961,-1096){\makebox(0,0)[b]{\smash{{\SetFigFont{9}{10.8}{\rmdefault}{\mddefault}{\updefault}{\color[rgb]{0,0,0}$3$}%
}}}}
\put(3781,-1096){\makebox(0,0)[b]{\smash{{\SetFigFont{9}{10.8}{\rmdefault}{\mddefault}{\updefault}{\color[rgb]{0,0,0}$2$}%
}}}}
\put(3601,-1096){\makebox(0,0)[b]{\smash{{\SetFigFont{9}{10.8}{\rmdefault}{\mddefault}{\updefault}{\color[rgb]{0,0,0}$1$}%
}}}}
\put(11116,-1906){\makebox(0,0)[b]{\smash{{\SetFigFont{9}{10.8}{\rmdefault}{\mddefault}{\updefault}{\color[rgb]{0,0,0}$n$}%
}}}}
\put(10486,-1906){\makebox(0,0)[b]{\smash{{\SetFigFont{9}{10.8}{\rmdefault}{\mddefault}{\updefault}{\color[rgb]{0,0,0}$2$}%
}}}}
\put(10306,-1906){\makebox(0,0)[b]{\smash{{\SetFigFont{9}{10.8}{\rmdefault}{\mddefault}{\updefault}{\color[rgb]{0,0,0}$1$}%
}}}}
\put(10666,-1906){\makebox(0,0)[b]{\smash{{\SetFigFont{9}{10.8}{\rmdefault}{\mddefault}{\updefault}{\color[rgb]{0,0,0}$3$}%
}}}}
\end{picture}%

%% file: Case-1-2-2.pdf_t
\begin{picture}(0,0)%
\includegraphics{Case-1-2-2.pdf}%
\end{picture}%
\setlength{\unitlength}{3522sp}%
\begingroup\makeatletter\ifx\SetFigFont\undefined%
\gdef\SetFigFont#1#2#3#4#5{%
  \reset@font\fontsize{#1}{#2pt}%
  \fontfamily{#3}\fontseries{#4}\fontshape{#5}%
  \selectfont}%
\fi\endgroup%
\begin{picture}(8306,2499)(2652,-2413)
\put(10531,-1411){\makebox(0,0)[b]{\smash{{\SetFigFont{9}{10.8}{\rmdefault}{\mddefault}{\updefault}{\color[rgb]{0,0,0}$C$}%
}}}}
\put(4771,-1096){\makebox(0,0)[b]{\smash{{\SetFigFont{9}{10.8}{\rmdefault}{\mddefault}{\updefault}{\color[rgb]{0,0,0}$n$}%
}}}}
\put(3961,-1096){\makebox(0,0)[b]{\smash{{\SetFigFont{9}{10.8}{\rmdefault}{\mddefault}{\updefault}{\color[rgb]{0,0,0}$3$}%
}}}}
\put(3781,-1096){\makebox(0,0)[b]{\smash{{\SetFigFont{9}{10.8}{\rmdefault}{\mddefault}{\updefault}{\color[rgb]{0,0,0}$2$}%
}}}}
\put(3601,-1096){\makebox(0,0)[b]{\smash{{\SetFigFont{9}{10.8}{\rmdefault}{\mddefault}{\updefault}{\color[rgb]{0,0,0}$1$}%
}}}}
\put(10936,-916){\makebox(0,0)[b]{\smash{{\SetFigFont{9}{10.8}{\rmdefault}{\mddefault}{\updefault}{\color[rgb]{0,0,0}$n$}%
}}}}
\put(10306,-916){\makebox(0,0)[b]{\smash{{\SetFigFont{9}{10.8}{\rmdefault}{\mddefault}{\updefault}{\color[rgb]{0,0,0}$2$}%
}}}}
\put(10126,-916){\makebox(0,0)[b]{\smash{{\SetFigFont{9}{10.8}{\rmdefault}{\mddefault}{\updefault}{\color[rgb]{0,0,0}$1$}%
}}}}
\put(10486,-916){\makebox(0,0)[b]{\smash{{\SetFigFont{9}{10.8}{\rmdefault}{\mddefault}{\updefault}{\color[rgb]{0,0,0}$3$}%
}}}}
\put(4096,-151){\makebox(0,0)[b]{\smash{{\SetFigFont{10}{12.0}{\rmdefault}{\mddefault}{\updefault}$T$}}}}
\put(6796,-151){\makebox(0,0)[b]{\smash{{\SetFigFont{10}{12.0}{\rmdefault}{\mddefault}{\updefault}$T'$}}}}
\put(2791,-736){\makebox(0,0)[b]{\smash{{\SetFigFont{9}{10.8}{\rmdefault}{\mddefault}{\updefault}(i)}}}}
\put(7381,-1096){\makebox(0,0)[b]{\smash{{\SetFigFont{9}{10.8}{\rmdefault}{\mddefault}{\updefault}{\color[rgb]{0,0,0}$n$}%
}}}}
\put(6571,-1096){\makebox(0,0)[b]{\smash{{\SetFigFont{9}{10.8}{\rmdefault}{\mddefault}{\updefault}{\color[rgb]{0,0,0}$3$}%
}}}}
\put(6391,-1096){\makebox(0,0)[b]{\smash{{\SetFigFont{9}{10.8}{\rmdefault}{\mddefault}{\updefault}{\color[rgb]{0,0,0}$2$}%
}}}}
\put(6211,-1096){\makebox(0,0)[b]{\smash{{\SetFigFont{9}{10.8}{\rmdefault}{\mddefault}{\updefault}{\color[rgb]{0,0,0}$1$}%
}}}}
\put(8911,-151){\makebox(0,0)[b]{\smash{{\SetFigFont{10}{12.0}{\rmdefault}{\mddefault}{\updefault}$F$}}}}
\put(2791,-1726){\makebox(0,0)[b]{\smash{{\SetFigFont{9}{10.8}{\rmdefault}{\mddefault}{\updefault}(ii)}}}}
\put(7381,-2086){\makebox(0,0)[b]{\smash{{\SetFigFont{9}{10.8}{\rmdefault}{\mddefault}{\updefault}{\color[rgb]{0,0,0}$n$}%
}}}}
\put(6571,-2086){\makebox(0,0)[b]{\smash{{\SetFigFont{9}{10.8}{\rmdefault}{\mddefault}{\updefault}{\color[rgb]{0,0,0}$3$}%
}}}}
\put(6391,-2086){\makebox(0,0)[b]{\smash{{\SetFigFont{9}{10.8}{\rmdefault}{\mddefault}{\updefault}{\color[rgb]{0,0,0}$2$}%
}}}}
\put(6211,-2086){\makebox(0,0)[b]{\smash{{\SetFigFont{9}{10.8}{\rmdefault}{\mddefault}{\updefault}{\color[rgb]{0,0,0}$1$}%
}}}}
\put(8371,-961){\makebox(0,0)[b]{\smash{{\SetFigFont{9}{10.8}{\rmdefault}{\mddefault}{\updefault}{\color[rgb]{0,0,0}$2$}%
}}}}
\put(8371,-556){\makebox(0,0)[b]{\smash{{\SetFigFont{9}{10.8}{\rmdefault}{\mddefault}{\updefault}{\color[rgb]{0,0,0}$B_1$}%
}}}}
\put(4771,-2086){\makebox(0,0)[b]{\smash{{\SetFigFont{9}{10.8}{\rmdefault}{\mddefault}{\updefault}{\color[rgb]{0,0,0}$n$}%
}}}}
\put(3961,-2086){\makebox(0,0)[b]{\smash{{\SetFigFont{9}{10.8}{\rmdefault}{\mddefault}{\updefault}{\color[rgb]{0,0,0}$3$}%
}}}}
\put(3781,-2086){\makebox(0,0)[b]{\smash{{\SetFigFont{9}{10.8}{\rmdefault}{\mddefault}{\updefault}{\color[rgb]{0,0,0}$2$}%
}}}}
\put(3601,-2086){\makebox(0,0)[b]{\smash{{\SetFigFont{9}{10.8}{\rmdefault}{\mddefault}{\updefault}{\color[rgb]{0,0,0}$1$}%
}}}}
\put(8453,-1747){\makebox(0,0)[lb]{\smash{{\SetFigFont{9}{10.8}{\rmdefault}{\mddefault}{\updefault}{\color[rgb]{0,0,0}$2$}%
}}}}
\put(8416,-1411){\makebox(0,0)[b]{\smash{{\SetFigFont{9}{10.8}{\rmdefault}{\mddefault}{\updefault}{\color[rgb]{0,0,0}$B_1$}%
}}}}
\put(8866,-1411){\makebox(0,0)[b]{\smash{{\SetFigFont{9}{10.8}{\rmdefault}{\mddefault}{\updefault}{\color[rgb]{0,0,0}$B_i$}%
}}}}
\put(8911,-961){\makebox(0,0)[b]{\smash{{\SetFigFont{9}{10.8}{\rmdefault}{\mddefault}{\updefault}{\color[rgb]{0,0,0}$1$}%
}}}}
\put(9091,-961){\makebox(0,0)[b]{\smash{{\SetFigFont{9}{10.8}{\rmdefault}{\mddefault}{\updefault}{\color[rgb]{0,0,0}$3$}%
}}}}
\put(9451,-961){\makebox(0,0)[b]{\smash{{\SetFigFont{9}{10.8}{\rmdefault}{\mddefault}{\updefault}{\color[rgb]{0,0,0}$n$}%
}}}}
\put(10306,-151){\makebox(0,0)[b]{\smash{{\SetFigFont{10}{12.0}{\rmdefault}{\mddefault}{\updefault}$F^*$}}}}
\put(10936,-1906){\makebox(0,0)[b]{\smash{{\SetFigFont{9}{10.8}{\rmdefault}{\mddefault}{\updefault}{\color[rgb]{0,0,0}$n$}%
}}}}
\put(10306,-1906){\makebox(0,0)[b]{\smash{{\SetFigFont{9}{10.8}{\rmdefault}{\mddefault}{\updefault}{\color[rgb]{0,0,0}$2$}%
}}}}
\put(10126,-1906){\makebox(0,0)[b]{\smash{{\SetFigFont{9}{10.8}{\rmdefault}{\mddefault}{\updefault}{\color[rgb]{0,0,0}$1$}%
}}}}
\put(10486,-1906){\makebox(0,0)[b]{\smash{{\SetFigFont{9}{10.8}{\rmdefault}{\mddefault}{\updefault}{\color[rgb]{0,0,0}$3$}%
}}}}
\put(9811,-1411){\makebox(0,0)[b]{\smash{{\SetFigFont{9}{10.8}{\rmdefault}{\mddefault}{\updefault}{\color[rgb]{0,0,0}$B_1\backslash C$}%
}}}}
\put(10531,-556){\makebox(0,0)[b]{\smash{{\SetFigFont{9}{10.8}{\rmdefault}{\mddefault}{\updefault}{\color[rgb]{0,0,0}$C$}%
}}}}
\end{picture}%

%% file: Case-2.pdf_t
\begin{picture}(0,0)%
\includegraphics{Case-2.pdf}%
\end{picture}%
\setlength{\unitlength}{3522sp}%
\begingroup\makeatletter\ifx\SetFigFont\undefined%
\gdef\SetFigFont#1#2#3#4#5{%
  \reset@font\fontsize{#1}{#2pt}%
  \fontfamily{#3}\fontseries{#4}\fontshape{#5}%
  \selectfont}%
\fi\endgroup%
\begin{picture}(8599,3309)(2629,-3223)
\put(7651,-1591){\makebox(0,0)[lb]{\smash{{\SetFigFont{9}{10.8}{\rmdefault}{\mddefault}{\updefault}{\color[rgb]{0,0,0}$n$}%
}}}}
\put(9001,-1951){\makebox(0,0)[b]{\smash{{\SetFigFont{9}{10.8}{\rmdefault}{\mddefault}{\updefault}{\color[rgb]{0,0,0}$1$}%
}}}}
\put(8326,-556){\makebox(0,0)[b]{\smash{{\SetFigFont{9}{10.8}{\rmdefault}{\mddefault}{\updefault}{\color[rgb]{0,0,0}$B$}%
}}}}
\put(4771,-2086){\makebox(0,0)[b]{\smash{{\SetFigFont{9}{10.8}{\rmdefault}{\mddefault}{\updefault}{\color[rgb]{0,0,0}$n$}%
}}}}
\put(3961,-2086){\makebox(0,0)[b]{\smash{{\SetFigFont{9}{10.8}{\rmdefault}{\mddefault}{\updefault}{\color[rgb]{0,0,0}$3$}%
}}}}
\put(3781,-2086){\makebox(0,0)[b]{\smash{{\SetFigFont{9}{10.8}{\rmdefault}{\mddefault}{\updefault}{\color[rgb]{0,0,0}$2$}%
}}}}
\put(3601,-2086){\makebox(0,0)[b]{\smash{{\SetFigFont{9}{10.8}{\rmdefault}{\mddefault}{\updefault}{\color[rgb]{0,0,0}$1$}%
}}}}
\put(8551,-1636){\makebox(0,0)[lb]{\smash{{\SetFigFont{9}{10.8}{\rmdefault}{\mddefault}{\updefault}{\color[rgb]{0,0,0}$2$}%
}}}}
\put(8551,-1816){\makebox(0,0)[lb]{\smash{{\SetFigFont{9}{10.8}{\rmdefault}{\mddefault}{\updefault}{\color[rgb]{0,0,0}$n$}%
}}}}
\put(9181,-1951){\makebox(0,0)[b]{\smash{{\SetFigFont{9}{10.8}{\rmdefault}{\mddefault}{\updefault}{\color[rgb]{0,0,0}$3$}%
}}}}
\put(9541,-1951){\makebox(0,0)[b]{\smash{{\SetFigFont{9}{10.8}{\rmdefault}{\mddefault}{\updefault}{\color[rgb]{0,0,0}$n-1$}%
}}}}
\put(7651,-1861){\makebox(0,0)[lb]{\smash{{\SetFigFont{9}{10.8}{\rmdefault}{\mddefault}{\updefault}{\color[rgb]{0,0,0}$n-1$}%
}}}}
\put(4096,-151){\makebox(0,0)[b]{\smash{{\SetFigFont{10}{12.0}{\rmdefault}{\mddefault}{\updefault}$T$}}}}
\put(6796,-151){\makebox(0,0)[b]{\smash{{\SetFigFont{10}{12.0}{\rmdefault}{\mddefault}{\updefault}$T'$}}}}
\put(2791,-736){\makebox(0,0)[b]{\smash{{\SetFigFont{9}{10.8}{\rmdefault}{\mddefault}{\updefault}(i)}}}}
\put(8911,-151){\makebox(0,0)[b]{\smash{{\SetFigFont{10}{12.0}{\rmdefault}{\mddefault}{\updefault}$F$}}}}
\put(2791,-1726){\makebox(0,0)[b]{\smash{{\SetFigFont{9}{10.8}{\rmdefault}{\mddefault}{\updefault}(ii)}}}}
\put(10576,-151){\makebox(0,0)[b]{\smash{{\SetFigFont{10}{12.0}{\rmdefault}{\mddefault}{\updefault}$F^*$}}}}
\put(6391,-2086){\makebox(0,0)[b]{\smash{{\SetFigFont{9}{10.8}{\rmdefault}{\mddefault}{\updefault}{\color[rgb]{0,0,0}$2$}%
}}}}
\put(6211,-2086){\makebox(0,0)[b]{\smash{{\SetFigFont{9}{10.8}{\rmdefault}{\mddefault}{\updefault}{\color[rgb]{0,0,0}$1$}%
}}}}
\put(7201,-2086){\makebox(0,0)[b]{\smash{{\SetFigFont{9}{10.8}{\rmdefault}{\mddefault}{\updefault}{\color[rgb]{0,0,0}$n-2$}%
}}}}
\put(8326,-1411){\makebox(0,0)[b]{\smash{{\SetFigFont{9}{10.8}{\rmdefault}{\mddefault}{\updefault}{\color[rgb]{0,0,0}$B$}%
}}}}
\put(11206,-1906){\makebox(0,0)[b]{\smash{{\SetFigFont{9}{10.8}{\rmdefault}{\mddefault}{\updefault}{\color[rgb]{0,0,0}$n$}%
}}}}
\put(10576,-1906){\makebox(0,0)[b]{\smash{{\SetFigFont{9}{10.8}{\rmdefault}{\mddefault}{\updefault}{\color[rgb]{0,0,0}$2$}%
}}}}
\put(10396,-1906){\makebox(0,0)[b]{\smash{{\SetFigFont{9}{10.8}{\rmdefault}{\mddefault}{\updefault}{\color[rgb]{0,0,0}$1$}%
}}}}
\put(10756,-1906){\makebox(0,0)[b]{\smash{{\SetFigFont{9}{10.8}{\rmdefault}{\mddefault}{\updefault}{\color[rgb]{0,0,0}$3$}%
}}}}
\put(10801,-556){\makebox(0,0)[b]{\smash{{\SetFigFont{9}{10.8}{\rmdefault}{\mddefault}{\updefault}{\color[rgb]{0,0,0}$C$}%
}}}}
\put(11206,-916){\makebox(0,0)[b]{\smash{{\SetFigFont{9}{10.8}{\rmdefault}{\mddefault}{\updefault}{\color[rgb]{0,0,0}$n$}%
}}}}
\put(10576,-916){\makebox(0,0)[b]{\smash{{\SetFigFont{9}{10.8}{\rmdefault}{\mddefault}{\updefault}{\color[rgb]{0,0,0}$2$}%
}}}}
\put(10396,-916){\makebox(0,0)[b]{\smash{{\SetFigFont{9}{10.8}{\rmdefault}{\mddefault}{\updefault}{\color[rgb]{0,0,0}$1$}%
}}}}
\put(10756,-916){\makebox(0,0)[b]{\smash{{\SetFigFont{9}{10.8}{\rmdefault}{\mddefault}{\updefault}{\color[rgb]{0,0,0}$3$}%
}}}}
\put(10801,-1501){\makebox(0,0)[b]{\smash{{\SetFigFont{9}{10.8}{\rmdefault}{\mddefault}{\updefault}{\color[rgb]{0,0,0}$C$}%
}}}}
\put(10036,-1501){\makebox(0,0)[b]{\smash{{\SetFigFont{9}{10.8}{\rmdefault}{\mddefault}{\updefault}{\color[rgb]{0,0,0}$B\backslash C$}%
}}}}
\put(2791,-2716){\makebox(0,0)[b]{\smash{{\SetFigFont{9}{10.8}{\rmdefault}{\mddefault}{\updefault}(iii)}}}}
\put(3781,-3076){\makebox(0,0)[b]{\smash{{\SetFigFont{9}{10.8}{\rmdefault}{\mddefault}{\updefault}{\color[rgb]{0,0,0}$2$}%
}}}}
\put(3601,-3076){\makebox(0,0)[b]{\smash{{\SetFigFont{9}{10.8}{\rmdefault}{\mddefault}{\updefault}{\color[rgb]{0,0,0}$1$}%
}}}}
\put(4771,-1096){\makebox(0,0)[b]{\smash{{\SetFigFont{9}{10.8}{\rmdefault}{\mddefault}{\updefault}{\color[rgb]{0,0,0}$n$}%
}}}}
\put(3961,-1096){\makebox(0,0)[b]{\smash{{\SetFigFont{9}{10.8}{\rmdefault}{\mddefault}{\updefault}{\color[rgb]{0,0,0}$3$}%
}}}}
\put(3781,-1096){\makebox(0,0)[b]{\smash{{\SetFigFont{9}{10.8}{\rmdefault}{\mddefault}{\updefault}{\color[rgb]{0,0,0}$2$}%
}}}}
\put(3601,-1096){\makebox(0,0)[b]{\smash{{\SetFigFont{9}{10.8}{\rmdefault}{\mddefault}{\updefault}{\color[rgb]{0,0,0}$1$}%
}}}}
\put(6391,-1096){\makebox(0,0)[b]{\smash{{\SetFigFont{9}{10.8}{\rmdefault}{\mddefault}{\updefault}{\color[rgb]{0,0,0}$2$}%
}}}}
\put(6211,-1096){\makebox(0,0)[b]{\smash{{\SetFigFont{9}{10.8}{\rmdefault}{\mddefault}{\updefault}{\color[rgb]{0,0,0}$1$}%
}}}}
\put(7201,-1096){\makebox(0,0)[b]{\smash{{\SetFigFont{9}{10.8}{\rmdefault}{\mddefault}{\updefault}{\color[rgb]{0,0,0}$n-2$}%
}}}}
\put(6436,-2581){\makebox(0,0)[lb]{\smash{{\SetFigFont{9}{10.8}{\rmdefault}{\mddefault}{\updefault}{\color[rgb]{0,0,0}$2$}%
}}}}
\put(6436,-2851){\makebox(0,0)[lb]{\smash{{\SetFigFont{9}{10.8}{\rmdefault}{\mddefault}{\updefault}{\color[rgb]{0,0,0}$1$}%
}}}}
\put(8911,-961){\makebox(0,0)[b]{\smash{{\SetFigFont{9}{10.8}{\rmdefault}{\mddefault}{\updefault}{\color[rgb]{0,0,0}$1$}%
}}}}
\put(9091,-961){\makebox(0,0)[b]{\smash{{\SetFigFont{9}{10.8}{\rmdefault}{\mddefault}{\updefault}{\color[rgb]{0,0,0}$2$}%
}}}}
\put(9451,-961){\makebox(0,0)[b]{\smash{{\SetFigFont{9}{10.8}{\rmdefault}{\mddefault}{\updefault}{\color[rgb]{0,0,0}$n$}%
}}}}
\put(8911,-2941){\makebox(0,0)[b]{\smash{{\SetFigFont{9}{10.8}{\rmdefault}{\mddefault}{\updefault}{\color[rgb]{0,0,0}$1$}%
}}}}
\put(10261,-2626){\makebox(0,0)[lb]{\smash{{\SetFigFont{9}{10.8}{\rmdefault}{\mddefault}{\updefault}{\color[rgb]{0,0,0}$2$}%
}}}}
\put(10261,-2806){\makebox(0,0)[lb]{\smash{{\SetFigFont{9}{10.8}{\rmdefault}{\mddefault}{\updefault}{\color[rgb]{0,0,0}$1$}%
}}}}
\put(8416,-2941){\makebox(0,0)[b]{\smash{{\SetFigFont{9}{10.8}{\rmdefault}{\mddefault}{\updefault}{\color[rgb]{0,0,0}$2$}%
}}}}
\put(8416,-2536){\makebox(0,0)[b]{\smash{{\SetFigFont{9}{10.8}{\rmdefault}{\mddefault}{\updefault}{\color[rgb]{0,0,0}$B$}%
}}}}
\put(7651,-601){\makebox(0,0)[lb]{\smash{{\SetFigFont{9}{10.8}{\rmdefault}{\mddefault}{\updefault}{\color[rgb]{0,0,0}$n$}%
}}}}
\put(7651,-871){\makebox(0,0)[lb]{\smash{{\SetFigFont{9}{10.8}{\rmdefault}{\mddefault}{\updefault}{\color[rgb]{0,0,0}$n-1$}%
}}}}
\end{picture}%